\documentclass[11pt]{article}

\usepackage{amsmath,amsfonts,amssymb,amsthm}
\usepackage{amsthm}
\usepackage{graphicx,color}
\usepackage{framed}
\usepackage{enumitem}
\usepackage{thmtools}
\usepackage{thm-restate}
\usepackage{xspace}
\usepackage{bm}
\usepackage{todonotes}
\usepackage{footnote}
\usepackage{mathtools}
\usepackage{mathrsfs}
\usepackage{todonotes}
\usepackage[most]{tcolorbox}
\usepackage{footnote}
\usepackage{rotating}
\usepackage[pdftex, plainpages = false, pdfpagelabels, 
bookmarks=true,
bookmarksopen = true,
bookmarksnumbered = true,
breaklinks = true,
linktocpage,
pagebackref,
colorlinks = true,  
linkcolor = blue,
urlcolor  = blue,
citecolor = red,
anchorcolor = green,
hyperindex = true,
hyperfigures
]{hyperref} 
\usepackage{tabularx}
\usepackage{tikz} 
\usetikzlibrary{calc}
\usepackage{thm-autoref}
\usepackage[nameinlink]{cleveref}

\newtheorem{theorem}{Theorem}
\newtheorem{lemma}{Lemma}
\newtheorem{claim}{Claim}
\newtheorem{corollary}{Corollary}
\newtheorem{definition}{Definition}
\newtheorem{observation}{Observation}
\newtheorem{proposition}{Proposition}

\theoremstyle{definition}

\theoremstyle{definition}

\usepackage[left=2cm, right=2cm, top=2cm]{geometry}

\newcommand{\NP}{\ensuremath{\mathsf{NP}}\xspace}

\newcommand{\WOH}{\ensuremath{\mathsf{W[1]}}-hard\xspace}
\newcommand{\WO}{\ensuremath{\mathsf{W[1]}}\xspace}

\newcommand{\FPT}{\ensuremath{\mathsf{FPT}}\xspace}

\let\oldlambda\lambda
\renewcommand{\lambda}{\ensuremath{\oldlambda}\xspace}
\let\oldalpha\alpha
\renewcommand{\alpha}{\ensuremath{\oldalpha}\xspace}
\let\oldDelta\Delta
\renewcommand{\Delta}{\ensuremath{\oldDelta}\xspace}

\renewcommand{\AA}{\ensuremath{\mathcal A}\xspace}
\newcommand{\BB}{\ensuremath{\mathcal B}\xspace}
\newcommand{\CC}{\ensuremath{\mathcal C}\xspace}
\newcommand{\DD}{\ensuremath{\mathcal D}\xspace}

\newcommand{\GG}{\ensuremath{\mathcal G}\xspace}

\let\mydelta\delta
\renewcommand{\delta}{\ensuremath{\mydelta}\xspace}

\let\mytau\tau
\renewcommand{\tau}{\ensuremath{\mytau}\xspace}

\let\mytheta\theta
\renewcommand{\theta}{\ensuremath{\mytheta}\xspace}

\let\mygamma\gamma
\renewcommand{\gamma}{\ensuremath{\mygamma}\xspace}

\let\myGamma\Gamma
\renewcommand{\Gamma}{\ensuremath{\myGamma}\xspace}
\usepackage{tikz}
\usetikzlibrary{arrows,positioning}

\newcommand{\cO}{\ensuremath{\mathcal{O}}\xspace}

\newcommand{\fb}{\textup{\textsc{Fair Bisection}}\xspace}
\newcommand{\sr}{\ensuremath{r^\circ}\xspace}

\newcommand{\cT}{\ensuremath{\mathcal{T}}\xspace}
\newcommand{\cA}{\ensuremath{\mathcal{A}}\xspace}
\newcommand{\mb}{\textsc{Minimum Bisection}\xspace}
\newcommand{\Oh}{\mathcal{O}}
\newcommand{\polyn}{n^{\Oh(1)}}

\newcommand{\true}{\textup{{\sf True}}\xspace}
\newcommand{\false}{\textup{{\sf False}}\xspace}
\newcommand{\child}{\textup{{\sf child}}\xspace}
\newcommand{\hgt}{\mathsf{ht}}
\newcommand{\lr}[1]{\left(#1\right)}
\newcommand{\LR}[1]{\left\{#1\right\}}
\newcommand{\faml}{\textup{\textsf{fam}}\xspace}
\newcommand{\csp}{\textup{\textsc{Binary Constrainted Satisfaction Problem}}\xspace}
\newcommand{\mdss}{\textup{\textsc{Multi-Dimensional Subset Sum}}\xspace}
\newcommand{\mdp}{\textup{\textsc{Multi-Dimensional Partition}}\xspace}
\newcommand{\cC}{\mathcal{C}}
\newcommand{\cV}{\mathcal{V}}
\newcommand{\cU}{\mathcal{U}}
\newcommand{\bbZp}{\mathbb{Z}_{\ge 0}}

\newcommand{\fbtp}{\texttt{FindBalancedTP}}
\newcommand{\parent}{\mathscr{P}}

\title{Parameterized Complexity of Fair Bisection\\(FPT-Approximation meets Unbreakability)\thanks{T.\ Inamdar is supported by ERC research and innovation programme (grant agreeement no. 819416). D. Lokshtanov and V. Surianarayanan are supported by NSF award CCF-2008838. S.\ Saurabh is supported by ERC research and innovation programme (grant agreeement no. 819416) and Swarnajayanti Fellowship grant DST/SJF/MSA01/2017-18.}

\author{
	Tanmay Inamdar\thanks{
		Department of Informatics, University of Bergen, Norway.}
	\and
	Daniel Lokshtanov\thanks{University of California Santa Barbara, USA.}
	\and
	Saket Saurabh\addtocounter{footnote}{-2}\footnotemark{}\addtocounter{footnote}{2} \thanks{Institute of Mathematical Sciences, Chennai, India.}
	\and
	Vaishali Surianarayanan \addtocounter{footnote}{-3}\footnotemark{}}
} 
\date{}

\begin{document}

\maketitle

\begin{abstract}
In the Minimum Bisection problem input is a graph $G$ and the goal is to partition the vertex set into two parts $A$ and $B$, such that $||A|-|B|| \leq 1$ and the number $k$ of edges between $A$ and $B$ is minimized. 
The problem is known to be \NP-hard, and assuming the Unique Games Conjecture even \NP-hard to approximate within a constant factor [Khot and Vishnoi, J.ACM'15]. On the other hand, a $\cO(\log n)$-approximation algorithm [R\"{a}cke, STOC '08] and a 
parameterized algorithm [Cygan et al., ACM Transactions on Algorithms '20] running in time $k^{\cO(k)}n^{\cO(1)}$ is known. 

The Minimum Bisection problem can be viewed as a clustering problem where edges represent similarity and the task is to partition the vertices into two equally sized clusters while minimizing the number of pairs of similar objects that end up in different clusters. 
Motivated by a number of egregious examples of unfair bias in AI systems, many  fundamental clustering problems have been revisited and re-formulated to incorporate fairness constraints. 
In this paper we initiate the study of the Minimum Bisection problem with fairness constraints. 
Here the input is a graph $G$, positive integers $c$ and $k$, a function $\beta:V(G) \rightarrow \{1, \ldots, c\}$  that assigns a color $\beta(v)$ to each vertex $v$ in $G$, and $c$ integers $r_1,r_2,\cdots,r_c$. The goal is to partition the vertex set of $G$ into two almost-equal sized parts $A$ and $B$  with at most $k$ edges between them, such that for each color $i\in \{1, \ldots, c\}$, $A$ has exactly $r_i$ vertices of color $i$. 
Each color class corresponds to a group which we require the partition $(A, B)$ to treat fairly, and the constraints that $A$ has exactly $r_i$ vertices of color $i$ can be used to encode that no group is over- or under-represented in either of the two clusters. 

We first show that introducing fairness constraints appears to make the Minimum Bisection problem qualitatively harder. Specifically we show that unless \FPT{}=\WO the problem admits no $f(c)n^{\cO(1)}$ time algorithm even when $k=0$.
On the other hand, our main technical contribution shows that is that this hardness result is simply a consequence of the very strict requirement that each color class $i$ has {\em exactly} $r_i$ vertices in $A$. In particular we give an $f(k,c,\epsilon)n^{\cO(1)}$ time algorithm that finds a balanced partition $(A, B)$ with at most $k$ edges between them, such that for each color $i\in [c]$, there are at most $(1\pm \epsilon)r_i$ vertices of color $i$ in $A$. 

Our approximation algorithm is best viewed as a proof of concept that the technique introduced by [Lampis, ICALP '18] for obtaining \FPT-approximation algorithms for problems of bounded tree-width or clique-width can be efficiently exploited even on graphs of unbounded width.
The key insight is that the technique of Lampis is applicable on tree decompositions with unbreakable bags (as introduced in~[Cygan et al., SIAM Journal on Computing '14]).
An important ingredient of our approximation scheme is a combinatorial result that may be of independent interest, namely that for every $k$, every graph $G$ admits a tree decomposition with adhesions of size at most $\Oh(k)$, unbreakable bags, and logarithmic depth. 

\end{abstract}

\section{Introduction}
Clustering is one of the most fundamental problems in computer science. In a clustering problem, we are typically interested in dividing the given collection of data points into a group of \emph{clusters}, such that the set of data points belonging to each cluster are more ``similar'' to each other, as compared to the points belonging to other clusters. Depending on the specific setting and application, there are a number of ways to model this abstract task of clustering as a concrete mathematical problem. We refer the reader to surveys such as ~\cite{xu2015comprehensive,rokach2009survey,blomer2016theoretical} for a detailed background and literature on the topic. 

In one such model of clustering, the input is represented as a simple, undirected graph, and the existence of an edge between a pair of vertices denotes that the two vertices are related to, or similar to, each other. For example, this is how one models social networks as graphs \cite{otte2002social} -- the set of vertices corresponds to people, and an edge represents that the two people are friends with each other. In this setting, the classical \mb problem can be thought of as a clustering problem \cite{YanH94,ChenZJ05} -- we are interested in finding two size-balanced clusters of vertices, such that the number of edges going across the two clusters is minimized. More formally, in \mb problem, we are given a graph $G = (V, E)$ on $n$ vertices, and a non-negative integer $k$, and the goal is to determine whether there exists a balanced edge cut $(A, B)$ of order $k$. Here, an {\em edge cut} $(A, B)$ is a partition of $V(G)$ into two non-empty subsets $A$ and $B$, an edge cut is {\em balanced} if $||A| - |B|| \leq 1$, and the {\em order} of the cut $(A, B)$ is the number of edges with one endpoint in $A$ and the other in $B$. The \NP-completeness of \mb has long been known \cite{GareyJ79}, and it is extensively studied from the perspective of approximation and parameterized algorithms. \mb admits a logarithmic approximation in polynomial time \cite{Racke08}, and it is hard to approximate within any constant factor, assuming the Unique Games Conjecture \cite{KhotV15}. In the realm of Parameterized Algorithms, one can solve the problem exactly in time $2^{\Oh(k \log k)} \cdot \polyn$, i.e., it is Fixed-Parameter Tractable (\FPT) parameterized by $k$ \cite{CyganLPPS19,CyganKLPPSW21}. 

More recently, the notion of \emph{fairness} has gained prominence in the literature of clustering algorithms -- and algorithm design in general. This is motivated from the fact that, often the real-life data reflects unconscious biases, and unless the algorithm is explicitly required to counteract these biases, the output of the algorithm may have real-life consequences that are \emph{unfair} ((see, e.g., \cite{grother2014face,obermeyer2019dissecting,Dastin.2017}). Researchers have proposed different models of fairness for the traditional center-based clustering problems, such as $k$-\textsc{Median/Means/Center}. These models of fairness can be broadly classified into two types -- \emph{individual fairness}, and \emph{group fairness}. At a high level, individual fairness requires that the solution treats each of the individuals (a point) in a fair way, e.g., every point has a cluster-center ``nearby''  \cite{ChenFLM19}. On the other hand, in the group fairness setting, the set of points is typically divided into multiple colors, where each color represents, say a particular demographic (such as gender, ethnicity etc.). In this setting, the fairness constraints are represented in terms of the colors as a group. There are multiple notions of group fairness (see, e.g., \cite{Bandyapadhyay0P19,ChenFLM19,FriedlerSV21,MakarychevV21,GhadiriSV21,JiaSS20}), but to the specific interest to us is the \emph{color-balanced clustering} model, studied in \cite{Schmidt19,HuangJV19,BandyapadhyayFS21}. Roughly speaking, in this setting we want the ``local proportions'' of all colors in every cluster to be approximately equal to their ``global proportions''.
Inspired from this \emph{color-balanced} notion of fairness, study the following \emph{fair} version of \mb.
\begin{tcolorbox}[colback=white!5!white,colframe=gray!75!black]
	\fb 
	\\\textbf{Input:} An instance $(G,c,k,\sr,\chi)$, where
	\begin{itemize}
		\item $G$ is an unweighted graph
		\item $c$ and $k$ are positive integers 
		\item $\chi:V(G)\rightarrow c$ is a coloring function on $V(G)$ using at most $c$ colors
		\item $\sr=(r_1,\cdots,r_c)$ is a $c$ length tuple of positive integers
	\end{itemize}
	\textbf{Question:} Does there exist an edge cut $(A,B)$ of $G$ of order at most $k$ having exactly $r_i$ vertices of color $i$ in $A$ for each $i\in [c]$. 
\end{tcolorbox}
%
%
In this problem formulation the color classes $i \in \{1, \ldots, c\}$ are protected groups which are required to be treated fairly by the clustering algorithm. The imposed fairness constraint for group $i$ is that, in the edge cut $(A, B)$, the set $A$ contains precisely $r_i$ vertices colored $i$.
We will say that an edge cut that satisfies the fairness constraints imposed by the tuple $\sr$ is $\sr$-{\em fair}. 
%
Thus, when $r_i$ is set to be precisely half of the number $c_i$ of vertices colored $i$ an $\sr$-fair edge cut must evenly split each color class across the two sides $A$ and $B$.

\paragraph*{Our Results.} %
It is quite easy to see that the existing parameterized algorithms~\cite{CyganKLPPSW21,CyganLPPS19} for \mb directly generalize to a $n^{O(c)}k^{O(k)}$ time algorithm for \fb{}\footnote{A formal proof of this claim is a corollary of our Theorem~\ref{thm:mainapprox}}. 
Therefore, the first natural question is whether it is possible to eliminate the dependence on $c$ in the exponent of $n$ in the running time. 
Our first result (Theorem~\ref{thm:fair-bisection-hardness}) is that, assuming $\FPT{}\neq\WO{}$, an $f(c)n^{O(1)}$ time algorithm is not possible even when $k=0$. In fact, this hardness result holds even in the special case where the vertices of each color are required to be evenly split across both partitions (in particular, when $2r_i = c_i$ for every $i$).

Our main technical contribution (Theorem~\ref{thm:mainapprox})  is to show that this hardness result is quite brittle. 
Indeed, the requirement that each color class $i$ have {\em exactly} $r_i$ vertices in $A$ is probably much too strong in the color-balanced fairness setting. 
We are satisfied even if the number of vertices of each color class is sufficiently close to the desired target number. 
We will say that an edge cut $(A, B)$ is $(\epsilon, \sr)$-{\em fair} if $A$ contains no more than $r_i(1+\epsilon)$ of vertices colored $i$ and $B$ contains no more than $(c_i-r_i)(1+\epsilon)$ vertices colored $i$.
We show (in Theorem~\ref{thm:mainapprox}) that there exists an algorithm that takes as input an instance $(G,c,k,\sr,\chi)$, together with an $\epsilon > 0$, runs in time $f(\epsilon, k, c)n^{\cO(1)}$, and if $G$ has a balanced \sr-fair edge cut $(\hat{A}, \hat{B})$ 



\paragraph*{Our Methods.}
The hardness result of Theorem~\ref{thm:fair-bisection-hardness} is a fairly straightforward parameterized reduction from {\sc Multi-Dimensional Subset Sum} parameterized by the dimension\footnote{The hardness of {\sc Multi-Dimensional Subset Sum} parameterized by the dimension is folklore, but we were unable to find a reference, so for completeness we provide a proof.}, whose main purpose is to put the parameterized approximation scheme of Theorem~\ref{thm:mainapprox} in context. We only discuss here the methods in the proof of of Theorem~\ref{thm:mainapprox}.

At a {\em very} high level the algorithm of  Theorem~\ref{thm:mainapprox} is the combination of two well-known techniques in parameterized algorithms: dynamic programming over tree decompositions with unbreakable bags (introduced by Cygan et al.~\cite{CyganLPPS19}), and the geometric rounding technique of Lampis~\cite{Lampis14} for parameterized approximation schemes for problems on graphs of bounded tree-width or clique-width. 
The conceptual novelty in (and perhaps the most interesting technical aspect of) our work is to realize that Lampis' technique can be applied even to dynamic programming algorithms over tree decompositions with unbounded width to yield approximation schemes for parameterized problems on general graphs.
Executing on this vision requires a few non-trivial technical insights, which we will shortly highlight. However, to describe these technical insights in more detail we first give a brief description of the two techniques that we combine.


\paragraph*{Lampis' Geometric Rounding Technique.}
We first discuss how the technique of Lampis~\cite{Lampis14} applies to tree decompositions of bounded width.
A {\em tree decomposition} of a graph $G$
is a pair $(T, \beta)$ where $T$ is a tree and $\beta$ is a function that assigns to each vertex $t \in V(T)$ a vertex set $\beta(t) \subseteq V(G)$ (called a {\em bag}) in $G$. 
To be a tree decomposition the pair $(T, \beta)$ must satisfy the tree-decomposition axioms: 
{\em (i)} for every $v \in V(G)$ the set $\{t \in V(T) ~:~ v \in \beta(t)\}$ induces a non-empty and connected subgraph of $T$, and
{\em (ii)} for every edge $uv \in E(G)$ there exists a $t \in V(T)$ such that $\{u,v\} \subseteq \beta(t)$.
The {\em width} (or tree-width) of a decomposition $(T, \beta)$ is defined as $\max_{t \in V(T)} |\beta(t)| - 1$.

Roughly speaking, Lampis' technique considers dynamic programming (DP) algorithms over a tree decomposition $(T, \beta)$ of $G$ of width $k$. 
In such an algorithm there is a DP-table for every node $t$ of the decomposition tree, and suppose that the entries in these tables are indexed by vectors in $\left\{1, 2, \ldots, n\right\}^d$ (for some integer $d$)
, where $n$ is the number of vertices of $G$. 
To decrease the size of the DP tables and thereby also the running time of the algorithm, one ``sparsifies'' the DP table to only consider entries in $S^d$, where $S = \left\{\lfloor(1+\delta)^i\rfloor : i \ge 0\right\}$. This makes the size of the DP table upper bounded by $(\log_{1+\delta} n)^{\Oh(d)}$, at the cost of introducing a multiplicative error of $(1+\delta)$ in every round of the DP algorithm (since now vectors in $\left\{1, 2, \ldots, n\right\}^d$ are ``approximated'' by their closest vector in $S^d$). If the decomposition tree $T$ has depth $\Oh(\log n)$ the dynamic program only needs $\Oh(\log n)$ rounds, and so the total error of the algorithm is a multiplicative factor of $(1+\delta)^{\Oh(\log n)}$. Setting $\delta = \epsilon/\log^2 n$ gives the desired trade-off between DP table size (and therefore running time) and accuracy. Luckily, every tree decomposition of width $k$ can be turned into a tree decomposition of width at most $3k+2$ and depth $\Oh(\log n)$~\cite{BodlaenderH98} and so this approach works on all graphs of tree-width $k$.

\paragraph{Tree Decompositions with Unbreakable Bags.} 
We now turn to the technique of Cygan et al.~\cite{CyganLPPS19} for \mb{}, namely dynamic programming over tree decompositions with small adhesions and unbreakable bags. We again need to define a few technical terms. An {\em adhesion} of a tree decomposition $(T, \beta)$ of a graph $G$ is a set $\beta(u) \cap \beta(v)$ for an edge $uv \in E(T)$. The {\em adhesion size} of a tree-decomposition is just the maximum size of an adhesion of the decomposition. 
A tree decomposition $(T, \beta)$ is said to have $(q,k)$-{\em unbreakable bags} if for every bag $\beta(t)$ of the decomposition and every edge-cut $(A, B)$ of order at most $k$ in $G$ it holds that $\min(|A \cap \beta(t)|, |B \cap \beta(t)|) \leq q$.

The main engine behind the algorithm of Cygan et al.~\cite{CyganLPPS19} (see also~\cite{CyganKLPPSW21}) is a structural theorem that for every graph $G$ and integer $k$ there exists a tree decomposition $(T, \beta)$ of $G$ with adhesion size at most $k$ and $(k+1, k)$-unbreakable bags. 
%
This is coupled with an observation that even though this tree decomposition might have unbounded tree-width, we can still do dynamic programming over this tree decomposition, keeping a DP table for every {\em adhesion} of the tree decomposition, rather than for every bag. 
However, while tree-width based DP algorithms utilize a simple recurrence to calculate the DP table at a bag from the tables of its children, Cygan et al.~\cite{CyganLPPS19} need to turn to a clever ``randomized contraction'' (see~\cite{ChitnisCHPP16}) based algorithm to compute the DP table for an adhesion from the DP tables of its children. 

\paragraph{Combining Tree Decompositions with Unbreakable Bags and Geometric Rounding.} 
As we mentioned eariler, the technique of Cygan et al.~\cite{CyganLPPS19} for \mb{} generalizes in a relatively straightforward way, to give a $f(k)n^{\cO(c)}$ time algorithm for \fb{}. Here we do dynamic programming over the tree decomposition of $G$ with adhesions of size $k$ and $(k+1, k)$-unbreakable bags. We have a DP table for every adhesion that is indexed by a vector in $[n]^c$ (this vector describes partial solutions, where the $i$'th element of the vector is the number of vertices of color $i$ that have so far been put on the $A$ side in this partial solution). 

We want to apply Lampis' geometric rounding technique and ``sparsify'' the DP table to only consider entries in $S^c$, where $S = \left\{\lfloor(1+\delta)^i\rfloor : i \ge 0\right\}$. 
There are a few technical obstacles to realizing this plan, that we overcome. The most important one of them is that the depth reduction theorem of Bodlaender and Hagerup~\cite{BodlaenderH98} only applies to tree decompositions of bounded width, therefore it is not immediate how to obtain a tree decomposition with small adhesions, unbreakable bags and logarithmic depth. 
A closer inspection of the proof sketch of Bodlaender and Hagerup~\cite{BodlaenderH98} reveals that a tree decomposition with adhesions of size $k$ and $(k+1,k)$-unbreakable bags can be turned into a 
tree decomposition with adhesions of size $\cO(k)$, and logarithmic depth, such that each bag of the new decomposition is the union of a constant number of bags of the old one (the bags in this new decomposition do {\em not} need to themselves be unbreakable). Nevertheless we prove that some careful modifications to this tree decomposition are sufficient to obtain a tree decomposition with adhesions of size $\cO(k)$, logarithmic depth, and $(\cO(k),k)$-unbreakable bags (see Theorem~\ref{theorem:low_depth_ub_td_poly}). 
We believe that Theorem~\ref{theorem:low_depth_ub_td_poly} will be a useful tool for future applications of Lampis' geometric rounding technique to tree decompositions with unbreakable bags.

\paragraph*{Organization of the Paper.} We begin by defining the basic notions on graphs and tree decompositions in \Cref{sec:prelims}. In \Cref{sec:ldtd}, we prove \Cref{theorem:lowdepth-near-unbreakable-treedecomp} that shows how to obtain logarithmic-depth unbreakable tree decompositions. Then, in \Cref{sec:algo}, we use such a tree decomposition to design our exact and approximate algorithms. In \Cref{sec:hardness}, we sketch the proof of our hardness result, which shows that \fb is \WOH parameterized by $c$ even when $k = 0$. Finally, in \Cref{sec:conclusion}, we give concluding remarks and future directions.


\section{Preliminaries} \label{sec:prelims}
For an integer $k$, we denote the set $\{1,2,\ldots,k\}$ by $[k]$.
For a graph $G$, {\em an edge cut} is a pair $A, B \subseteq V(G)$ such that $A\cup B=V(G)$ and $A\cap B=\emptyset$. The order of an edge cut $(A,B)$ is $|E(A,B)|$, that is, the number of edges with one endpoint in $A$ and the other in $B$. For a subset $X \subseteq V(G)$, let $G \setminus X$ denote the graph $G[V(G) \setminus X]$. For an edge cut $(A, B)$, and a subset $X \subseteq V(G)$, the cut \emph{induced} on $X$ by $(A, B)$ is $(A \cap X, B \cap X)$.

\begin{definition}[unbreakability]
A set $X \subseteq V(G)$ is \emph{$(q,s)$-edge-unbreakable} if every edge cut $(A,B)$ of order at most $s$
satisfies $|A \cap X| \leq q$ or $|B \cap X| \leq q$.
\end{definition}

For a rooted tree $T$ and vertex $t\in V(T)$, we denote by $T_t$ the subtree of $T$ rooted at $t$. For a rooted tree $T$ and a non-root vertex $t \in V(T)$, we denote the parent of $t$ by $\parent(t)$. The depth of a tree $T_t$ is the maximum length of a $t$ to leaf path in $T_t$. For a node $t$, we denote $\hgt_T(t)$ to be the the depth of the subtree $T_t$ rooted at $t$ in $T$.

Another important notion that we need is of tree decomposition where bags are ``highly connected'', i.e., unbreakable. Towards this we first define tree decomposition, tree-width and associated notions and notations that we make use of. For a rooted tree $T$ and vertex $v \in V(T)$ we denote by $T_v$ the subtree of $T$ rooted at $v$. We refer to the vertices of $T$ as nodes. 

A {\em tree decomposition} of a graph $G$ is a pair $(T,\beta)$ where $T$ is a rooted tree and $\beta$ is a function from $V(T)$ to $2^{V(G)}$ such that the following three conditions hold. 
\begin{description}
\setlength{\itemsep}{-2pt}
\item[(T1)] $\underset{t\in V(T)}{\bigcup}\beta(t) = V(G)$;
\item[(T2)] For every $uv\in E(G)$, there exists a node $t\in T$ such that $\beta(t)$ contains both $u$ and $v$; and 
\item[(T3)] For every vertex $u\in V(G)$, the set $T_u=\{t\in V(T): u\in \beta(t)\}$, i.e., the set of nodes whose corresponding bags contain $u$, induces a connected subtree of $T$. 
\end{description}
For every $t \in V(T)$ a set $\beta(t) \subseteq V(G)$, is called a \emph{bag}.  We can extend the function $\beta$ to subsets of $V(T)$ in the natural way: for a subset $X \subseteq V(T)$, $\beta(X) \coloneqq \bigcup_{x \in X} \beta(x)$. 

For $s,t \in V(T)$ we say that \emph{$s$ is a descendant of $t$}
or that \emph{$t$ is an ancestor of $s$} if $t$ lies on the unique path from $s$ to the root;
note that a node is both an ancestor and a descendant of itself.
%
By $\textsf{child}(t)$, we denote the set of children of $t$ in 
$T$. For any $X \subseteq V(T)$, define $G_X \coloneqq G[\cup_{t\in X}\beta(V(T_t))]$. 

We define an {\em adhesion} of an edge $e = (t, t_0) \in E(T)$ to be the set $\sigma(e) \coloneqq \beta(t) \cap \beta(t_0)$, and an adhesion of $t \in V (T)$ to be $\sigma(t) \coloneqq \sigma({t, \parent(t)})$, or $\sigma(t) = \emptyset$ if the parent of $t$ does not exist, i.e., when $t$ is the root of $T$.
We define the following notation for convenience:
$$\gamma(t) \coloneqq \bigcup_{s:\text{ descendant of }t} \beta(s)$$
$$\alpha(t) \coloneqq \gamma(t) \backslash \sigma(t), \qquad G_t \coloneqq G[\gamma(t)] - E(G[\sigma(t)]).$$

We say that a rooted tree decomposition $(T, \beta)$ of $G$ is \emph{compact} if for every node $t \in V (T)$ for which $\alpha(t) \neq \emptyset$ we have that $G[\alpha(t)]$ is
connected and $N_G(\alpha(t)) = \sigma(t)$.

\subsection{Splitters} 
We start by defining the notion of splitters. We will need this for our color coding based dynamic programming algorithm.

\begin{definition}[\cite{NaorSS95}]
\label{splitterdefn}
An $(n,k,\ell)$ splitter $\mathcal{F}$ is a family of functions from $[n]\rightarrow[\ell]$ such that for all  $S\subseteq [n]$, $|S|=k$, there exists a function $ f\in \cal F$ that splits $S$ evenly. That is, for all $1\leq j,j'\leq \ell$, $|f^{-1}(j')\cap S|$ and $|f'^{-1}(j)\cap S|$ differ by at most one.
\end{definition}

We will need following algorithm to compute splitters with desired parameters. 
\begin{theorem}[\cite{NaorSS95}]\label{splitter2}
For all $\ n,k\geq 1$ one can construct an $(n,k,k^2)$ splitter family of size $k^{\cO(1)}\log n$ in time $k^{\cO(1)}n\log n$.
\end{theorem}  

The next lemma is a simple application of Theorem~\ref{splitter2} from \cite{Lokshtanov0S20}, and is used as a subroutine in our algorithm for \fb. 
\begin{lemma}[\cite{Lokshtanov0S20}]
\label{splittersetlemma}
There exists an algorithm that takes as input a set $S$, two positive integers $s_1$ and $s_2$ that are less than $|S|$, and outputs a family $\mathcal{S}$ of subsets of $S$ having size $\cO((s_1+s_2)^{\cO(s_1)}\log |S|)$ such that for any two disjoint subsets $X_1$ and $X_2$ of $S$ of size at most $s_1$ and $s_2$, $\mathcal{S}$ contains a subset $X$ that satisfies $X_1\subseteq X$ and $X_2\cap X =\emptyset$ in time $\cO((s_1+s_2)^{\cO(s_1)}|S|^{\cO(1)})$.
\end{lemma}
\section{Obtaining a Low Depth Unbreakable Tree Decomposition} \label{sec:ldtd}
In this section we show that there exists a tree decomposition that has low (i.e., $\cO(\log n)$) depth, small-size (i.e., $\Oh(k)$) adhesions, and $(\Oh(k),k)$-unbreakable bags. To this end, we design a polynomial-time algorithm that, given a tree decomposition with small adhesions and unbreakable bags, produces a tree decomposition with the aforementioned properties. In the next section, we design a dynamic programming algorithm over such a low depth decomposition to obtain an \FPT approximation for \fb. 

In our algorithm, we use the notion of a \emph{tree partition} of a graph, which, informally speaking, capture the ``tree-likeness'' of a graph. Tree partitions were introduced by~\cite{Seese85,Halin91a}. This notion is easier to define and to think about, compared to tree decompositions. Although tree partitions are not as versatile as tree decompositions when it comes to algorithm design, there are a few circumstances---such as ours---where they are useful. A formal definition follows.
\begin{definition}[Tree Partition] \label{def:tree-partition} A tree partition of a graph $G$ is a pair $(\cT,\tau)$ where $\cT$ is a tree and $\tau: V(G)\rightarrow V(\cT)$ is a function from $V(G)$ to $V(\cT)$ such that for each $e=(u,v)\in E(G)$ either $\tau(u)=\tau(v)$ or $(\tau(u),\tau(v))\in E(\cT)$. A rooted tree partition $(\cT,\tau)$ with root $r$ is the tree partition $(\cT,\tau)$ where the tree $\cT$ is a rooted tree with root $r$.
\end{definition}
We remark that we use calligraphic font ($\cT$) to denote trees corresponding to Tree Partitions to easily distinguish them from graphs that are trees. 
Observe that for a tree $T$, the pair $(\cT=T,\tau)$ where $\tau(v)=v$ for each $v\in T$ is a trivial tree partition of $T$. 

For our result we only use tree partitions of trees. Given a tree decomposition $(T,\beta)$ with small adhesions and unbreakable bags, our goal in this section is to use $(T,\beta)$ to obtain a tree decomposition of bounded height without blowing up the adhesion size and unbreakability guarantees too much. For this we first find a tree partition $(\cT,\tau)$ of $T$ that satisfies additional properties, such as logarithmic depth and for each $t \in V(\cT)$, it holds that $|\tau^{-1}(t)|\leq 4$. Using this tree partition, we obtain a tree decomposition $(\mathcal{T},\beta_1)$ whose underlying tree is $\mathcal{T}$, and each bag $\beta_1(t)$,  $t\in V(\cT)$ is a union of at most $4$ bags of $(T,\beta)$; $\beta_1(t)=\bigcup_{x\in \tau^{-1}(t)}\beta(x)$. This tree decomposition already has bounded height, small adhesion size and each bag is a union of at most four unbreakable bags of $(T,\beta)$. From here with some extra work we obtain a tree decomposition with unbreakable bags as well.
For this we use other properties of $(\cT,\tau)$ to modify $(\cT,\beta_1)$ to obtain our desired tree decomposition.
 
As outlined above, for our result we need tree partitions of a tree satisfying some properties. We now define such tree partitions below and show how to find one in polynomial time.

\begin{definition}[Nice Tree Partition]\label{defn:nice_tree_partition}
A tree partition $(\mathcal{T},\tau)$ of a tree $T$ is said to be a nice tree partition if it satisfies the following properties:
\begin{enumerate}
    \item $\cT$ has depth at most $\lceil\log_2|V(T)|\rceil$ 
    \item for each $t\in V(\cT)$, $1<|\tau^{-1}(t)|\leq 4$.
    \item for each $t\in V(\cT)$, $T[V_t]$ is a subtree of $T$, where  $V_t = \bigcup_{x\in V(\cT_t)}\tau^{-1}(x)$.
\end{enumerate}
\end{definition}

We now show how to find a nice tree partition of a tree in polynomial time. For this we use a recursive procedure. The core idea in each recursive step is to map a balanced separator $b$ of the tree to the root of the tree partition. To ensure the connectivity properties of a tree partition, we have a set $M$ of marked vertices in the tree that are always mapped to the root of the tree partition in addition to $b$. Then for each connected component in the forest obtained by removing $M\cup \{b\}$ from the tree, we mark new vertices and recurse. We need some extra work to make sure that every node in the tree partition is mapped to by only a constant number of nodes in the tree. For this we ensure that in each recursive call we mark only a few ($\leq 2$) new vertices. 
\begin{lemma}\label{lemma:tp_tree}
Given a tree $T$ on $n$ vertices with root $r$, one can in polynomial time compute a rooted nice tree partition $(\cT,\tau)$ of $T$ with root $r_{\cT}$ such that $\tau(r)=r_{\cT}$. 
\end{lemma}
\begin{proof} 

We now design a procedure \fbtp\ that takes as argument a tree $T'$, and a non-empty set $M\subseteq V(T')$ of size at most $2$, and returns a rooted nice tree partition $(\cT',\tau')$ of $T'$ with root $r'$, such that $M\subseteq \tau'^{-1}(r')$.
%
We will invoke this procedure on the input tree $T$ with $M=\{r\}$, where $r$ is the root of $T$ to obtain a rooted nice tree partition $(\cT,\tau)$ of $T$.

In the procedure \fbtp$(T',M)$ we carry out the following steps: 
\begin{itemize}
    \item We find a balanced bisector $b$ of $T$ and initialize $M'=M\cup \{b\}$. 
    \item If all vertices in $M'$ do not lie on a path in $T'$, we add an extra vertex $x$ to $M'$. Let $x$ be the last common vertex on the path from $m_1$ to $m_2$ and the path from $m_1$ to $b$ in $T$. Modify $M'=M'\cup \{x\}$.
    \item For each tree $H$ in the forest $T'\setminus M'$, we recursively call \fbtp($H$,$M_H$) where $M_H=N_{T'}(M')\cap V(H)$ is the set of neighbors of vertices in $M'$ in $H$. Let $(\mathcal{H},\tau_H)$ be the tree partition returned by this procedure call. 
    \item We now construct a tree partition $(\cT',\tau')$ with root $r'$. We assign $\tau'^{-1}(r')=M'$. Then for each tree $H$ in the forest $T'\setminus M'$, we make $\mathcal{H}$ a subtree of $\cT'$ by attaching the root of $\mathcal{H}$ as a child to $r'$. Further for each $t\in V(\mathcal{H})$ we assign $\tau'^{-1}(t)=\tau_H(t)$. 
    \item We return $(\cT',\tau')$.
\end{itemize}

We now prove by induction on the number of vertices in the tree that, for any tree $T'$ with root $r'$, and any non-empty subset $M \subseteq V(T')$ with $0 < |M| \le 2$, the procedure \fbtp$(T',M)$ returns a rooted nice tree partition $(\cT',\tau')$ of $T'$ with root $r'$ such that $M\subseteq \tau'^{-1}(r')$.

\textit{Base Case $|V(T')|=1$ or $V(T')=M'$:} In this case $V(\cT')=\{r'\}$ and $\tau(r')=M'$. Observe that $0 < |M'| \le 4$. This is because the procedure is called with a non-empty set $M$ of size at most two. Then, the procedure initializes $M'=M$, and adds at most two other vertices ($b$ and $x$) in $V(T')$ to $M'$. Thus $(\cT',\tau')$ is a nice tree partition with root $r'$ and $M\subseteq \tau'^{-1}(r')$. 

Now we prove the {\em inductive case} where $V(T')$ has size $i$, $i>1$ and $V(T')\neq M'$. For this we assume the inductive hypothesis that the procedure returns a tree partition with the desired properties for all trees $H$ having less than $i$ vertices and non-empty sets $M'\subseteq V(H)$ of size at most two. 

Let $H$ be a tree in the forest $T'\setminus M'$. We now show that $|V(H)|\leq \lceil|V(T')|/2\rceil$ and $1<|M_H|\leq 2$. 
By construction $M'$ contains the vertex $b$, a balanced bisector of $T'$. Thus $V(H)$ has size at most $\lceil|V(T')|/2\rceil$. 
$|M_H|>1$ since $H$ contains at least one child of $M'$ since it is a tree in the forest $T'\setminus M'$. To show $|M_H|\leq 2$, we first show there is a vertex $s$ in $M'$ such that in the forest $T'\setminus \{s\}$ every vertex $s'\in M'\setminus\{s\}$ is contained in a different tree. If all vertices in $M\cup \{b\}$ do not lie on a path in $T'$, then $s$ is just the vertex $x$ we added to $M'$ in the second step of the procedure. If the vertices of $M\cup \{b\}$ lie on a path $P$ in $T'$ then $M'=M\cup \{b\}$. In this case if $|M'|\leq 2$, then $s$ is any vertex in $M'$. On the other hand if $|M'|=3$, then $s$ is the second vertex from $M'$ in the path $P$. Due to the property of $s$, $H$ may contain a child of $s$ and a child of one other $s'\in M'\setminus \{s\}$. Thus $|M_H|\leq 2$.

Since $H$ is a tree with $|V(H)|\leq \lceil|V(T')|/2\rceil$ and $1<|M_H|\leq 2$, by induction the tree partition $(\mathcal{H},\tau_H)$ returned by the call to the procedure \fbtp($H$,$M_H$) is a nice tree partition with root $r_H$ and $M_H\subseteq \tau_H^{-1}(r_H)$. 

We now show that $(\cT,\tau')$ is a rooted tree partition of $T'$ with root $r'$. 
First we show that each vertex $v\in \cT$ is mapped to exactly one vertex $t\in \cT$ by $\tau'$. If $v\in M'$, then $\tau(v)$ is mapped to $r'$. If $v\in H$, $H\in T'\setminus M'$, then since $(\mathcal{H},\tau_H)$ is a rooted tree partition of $H$, $\tau(v)=\tau_H(v)$ by construction. Next we show that each edge $(x,y)\in T'$ satisfies either $\tau'(x)=\tau'(y)$ or $(\tau'(x),\tau'(y))\in E(\cT)$, by considering three cases. (i) If $x,y\in M'$, then this is trivially true. (ii) If $x,y\notin M'$ then $x,y$ must belong to some tree $H\in T'\setminus M'$ and thus by induction $(\tau'(x)=\tau_H(x),\tau'(y)=\tau_H(y))\in E(\cT)$. (iii) If $x\in M'$ and $y\notin M'$, by construction, $\tau(x)=r'$ and $y\in M_H$ for some $H\in T'\setminus M'$. Since $M_H\subseteq \tau_H^{-1}(r_H)$ and $r_H$ is a child of $r'$ in $\cT'$, $(\tau'(x)=r,\tau'(y)=r_H)\in E(\cT)$.

$M\subseteq \tau'^{-1}(r')$ just by construction. We now prove properties $(1)-(4)$ in Definition~\ref{defn:nice_tree_partition} to show that $(\cT',\tau')$ is a nice tree partition. Recall that for each $H\in T'\setminus M'$, $H$ is a subtree of $T'$ with $r_H$ being a child of $r'$ in $\cT'$. 
Then, since $|V(H)| \le \frac{|V(T)|}{2}$, by inductive hypothesis, $\mathcal{H}$ has depth at most $\lceil\log_2(|V(T')|/2)\rceil = \lceil\log_2(|V(T')|)\rceil-1$. Therefore, $\cT'$ has depth $\lceil\log_2|V(T')|\rceil$, since the addition of the root $r'$ increases the depth by $1$.
For each $t'\in V(\cT')$, $1<|\tau'^{-1}(t')|\leq 4$ since $1<|\tau'^{-1}(r')|\leq 4$ and for each $H\in T'\setminus M'$ and for each $t\in V(\mathcal{H})$, $1<|\tau^{-1}_H(t)|\leq 4$. For each $t'\in V(\cT')$, $T'[V_{t'}]$ is a subtree of $T'$, where $V_{t'}=\bigcup_{x\in V(\cT'_{t'})}\tau'^{-1}(x)$ because for each $H\in T'\setminus M'$ and $t\in V(\mathcal{H})$, $H[V_t]$ is a subtree of $H$, where $V_t=\bigcup_{x\in V(\mathcal{H}_t)}\tau_H^{-1}(x)$ and $H$ is subtree of $T'$. 
%
%
This completes the proof.
\end{proof}
Let $(T,\beta)$ be a rooted tree decomposition of a graph $G$ with root $r$, $(q,k)$-unbreakable bags and adhesions of size at most $k$. Further let $(\cT,\tau)$ be a rooted nice tree partition of $T$ with root $r_{\cT}$ as provided by Lemma~\ref{lemma:tp_tree}. 
We now show that we can obtain a natural rooted tree decomposition $(\cT,\beta_1)$ of $G$ where the tree in the decomposition is $\cT$. Here $\beta_1:V(\cT)\rightarrow 2^{V(G)}$ and $\beta_1(t)=\bigcup_{x\in \tau^{-1}(t)} \beta(x)$. 

From now on we fix $G,(T,\beta),(\cT,\tau),$ and $\beta_1$ for the rest of the section. We remark that to prove $(\cT,\beta_1)$ is a tree decomposition we will not need the {\em nice} properties of $(\cT,\tau)$ nor the properties of the bags and adhesions in $(T,\beta)$. We will later use them to deduce some helpful structural properties of $(\cT,\beta_1)$. 

\Cref{fig:td_tp} is the accompanying figure for the proof of the following lemma.

\begin{lemma}\label{lemma:treedecomp_from_treepartition}
The pair $(\cT,\beta_1)$ is a tree decomposition of $G$.
\end{lemma}
\begin{proof}
We show that $(\cT,\beta_1)$ satisfies the three properties required for it to be a tree decomposition of $G$. We first show that for each vertex $v\in V(G)$, there is a bag $\beta_1(t)$, $t\in \cT$ containing $v$, i.e $v\in \beta_1(t)$. Since $(\cT,\tau)$ is a tree partition of $T$, every vertex in $T$ is mapped to a vertex in $\cT$ by $\tau$. So for each $x\in T$, by definition of $\beta_1$, $\beta(x)\subseteq \beta_1(\tau(x))$. Thus since $(T,\beta)$ is a tree decomposition of $G$, each vertex $v\in V(G)$ is contained in some bag $\beta(x)$, $x\in V(T)$ and thus $v\in\beta_1(t)$, where $t = \tau(x)$. 

Next we show that for each edge $(u,v)\in E(G)$, there exists a bag $\beta_1(t)$, $t\in \cT$ such that $u,v\in \beta_1(t)$. Since $(T,\beta)$ is a tree decomposition of $G$, there exists a bag $\beta(x)$, $x\in T$ such that $u,v\in \beta(x)$. Therefore by definition of $\beta_1$, $u,v\in \beta_1(t=\tau(x))$.  

Let $v\in V(G)$ and let $\beta^{-1}(v) = \{x \in V(T): v \in \beta(x)\}$ and $\beta_1^{-1}(v) = \{t \in V(\cT): v \in \beta_1(t)\}$. Since $(T,\beta)$ is a tree decomposition of $G$, the set $\beta^{-1}(v)$ induces a connected subgraph of $T$. We now show that $\beta_1^{-1}(v)$ is a connected subgraph of $\cT$. Suppose not, let $X_1$ be a connected component in $\beta_1^{-1}(v)$ such that there is no edge in $E(\cT)$ from $X_1$ to $\beta_1^{-1}(v)\setminus X_1$. Let $X=\bigcup_{t\in X_1}\{x:x\in \tau^{-1}(t)\}\cap \beta^{-1}(v)$. Since $\beta^{-1}(v)$ is connected, there is an edge $(u,v)$ in $E(T)$ with $u\in X$ and $v\in X\setminus \beta^{-1}(v)$. Observe that $\tau(u)\in X_1$ and $\tau(v)\in \beta_1^{-1}(v)\setminus X_1$ by definition of $X_1$ and $X$. Thus $\tau(u)\neq \tau(v)$ and since $(\cT,\tau)$ is a tree partition of $T$, $(\tau(u),\tau(v))$ is an edge in $\cT$. This contradicts our assumption that there is no edge from $X_1$ to $\beta_1^{-1}(v)\setminus X_1$ in $E(\cT)$ and proves that $\beta_1^{-1}(v)$ is connected. 
All three properties combined show that $(\cT,\beta_1)$ is a tree decomposition of $G$.
\end{proof}

Observe that since $(\cT,\tau)$ is a nice tree partition of $T$, the tree decomposition $(\cT,\beta_1)$ has depth at most $\lceil\log_2|V(G)|\rceil$ and each of its bags is a union of at most four bags of $(T,\beta)$. We now prove a few other useful properties of $(\cT,\beta_1)$ that will help us design our desired tree decomposition.

\begin{figure}
    \centering
    \includegraphics[page=1,scale=0.7]{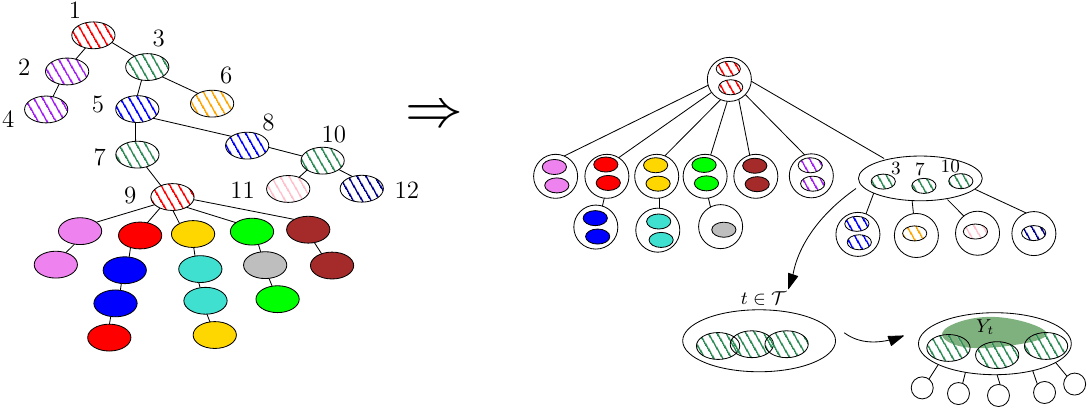}
    \caption{Left: Tree decomposition $(T,\beta)$ with bags colored for easy understanding. Right: ree decomposition $(\cT,\beta_1)$ that is constructed using a tree partition $(\cT,\tau)$. The bags in $(T,\beta)$ are mapped according to $\tau$ to $(\cT,\beta_1)$. The bags in $\tau^{-1}(t)$, $t\in \cT$ can overlap as demonstrated by bags $3,7,10$ but their overlap is small and contained in $Y_t$.}
    \label{fig:td_tp}
\end{figure}

\begin{figure}
    \centering
    \includegraphics[page=2,scale=1]{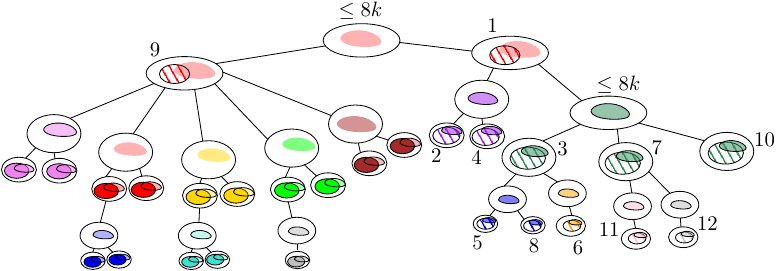}
    \caption{This shows the final tree decomposition $(T^*,\beta^*)$ constructed using $(T,\beta)$ and $(\cT,\beta_1)$ as shown in Fig~\ref{fig:td_tp}. $(T^*,\beta^*)$ is our desired tree decomposition with low depth, unbreakable bags and small adhesions.}
    \label{fig:final_td}
\end{figure}

\begin{lemma}\label{lemma:tree_partition_decomp_prop}
There exists a function $\gamma: V(\cT) \to V(T)$, and a set $Y_t \subseteq \beta_1(t)$ for each node $t \in V(\cT)$, that satisfy, for each node $t \in V(\cT)$, the following properties:
\begin{enumerate}
    \item $|Y_t|\leq 8k$ 
    \item If $t$ is not the root $r_{\cT}$ then $\beta_1(\parent(t))\cap \beta_1(t)\subseteq Y_t$
    \item for each distinct $x,y\in \tau^{-1}(t)$, $\beta(x)\cap \beta(j)\subseteq Y_t$ 
    \item for each child $t_c$ of $t$ in $\cT$, it holds that $t = \tau(\gamma(t_c))$ and $\beta_1(t_c)\cap \beta_1(t)\subseteq Y_t\cup \beta(\gamma(t_c))$
\end{enumerate}
Furthermore, $\gamma$ and the sets $Y_t$ for each $t \in V(\cT)$ can be computed in polynomial time.
\end{lemma}
\begin{proof}
For $e=(x,x')\in E(T)$, let $\sigma(e)=\beta(x)\cap\beta(x')$ be the adhesion of edge $e$ in $(T,\beta)$. 

For $r_{\cT}$, let $T_{r_{\cT}}=\bigcup_{x\in \tau^{-1}(r_\cT)} \sigma((x,\parent(x)))$ and $\gamma(r_{\cT})=r$.
For $t\in V(\cT)$, $t\neq r_{\cT}$, let $E_t=\{e:e=(x,y)\in E(T),x\in\tau^{-1}(\parent(t)), y\in \tau^{-1}(t)\}$. 
Then let $Y_t=\bigcup_{x\in \tau^{-1}(r_\cT)} \sigma((x,\parent(x))) \cup \bigcup_{e\in E_T}\sigma(e)$. 

For each $x$ in $\tau^{-1}(\parent(t))$, there exists at most one $y\in \tau^{-1}(t)$ such that $(x,y)\in E(T)$. This is because $T[V_t]$ is connected, where $t=\bigcup_{x\in V(\cT_t)}\tau^{-1}(x)$. If $x$ has an edge to two vertices in $\tau(t)$, then there would be a cycle in $T$. Thus, $|E_t|\leq 4$. Further since $\cT$ is from a nice tree partition, $|\tau^{-1}(t)|\leq 4|$. Therefore $|Y_t|\leq 8k$ for each $t\in V(\cT)$.

For $t\neq r_{\cT}$, $t\in V(\cT)$. Since $(T,\beta)$ is a tree decomposition and $E_t$ are the only set of edges in $\cT$ between vertices in $\tau^{-1}(t)$ and $\tau^{-1}(\parent(t))$. Further since $(\cT,\beta_1)$ is a tree decomposition with $\beta_1(t)=\bigcup_{x\in \tau^{-1}(t)} \beta(x)$, $\beta_1(\parent(t))\cap \beta_1(t)\subseteq Y_t$.

Recall that $(T, \beta)$ is a tree decomposition. Therefore, for any two nodes $x,y \in V(T)$, consider a vertex $v \in \beta(x) \cap \beta(y)$. Note that $v$ must belong to every bag of the node appearing on the unique $x$ to $y$ path in $T$. Let $z \in V(T)$ be the least common ancestor of $x$ and $y$ -- note that $z$ may be equal to $x$ or $y$ or neither. Suppose $z \not\in \LR{x, y}$. Then, $v$ must appear in $\beta(x) \cap \beta(\parent(x)) = \sigma(x, \parent(x))$, as well as in $\sigma(y, \parent(y))$. Otherwise, if $z = x$ (w.l.o.g.), then $v \in \sigma(y, \parent(y))$. Since $Y_t \supseteq \sigma(x, \parent(x)) \cup \sigma(y, \parent(y))$, we get the third property.

Let $t_c\in V(\cT)$, $t_c\neq r_{\cT}$. Further let $t=\parent(t_c)$ in $\cT$. We now show that for all but at most one vertex $x\in \tau^{-1}(t)$, $\beta_1(t_c)\cap \beta(x)\subseteq Y_t$. If for all $x\in \tau^{-1}$, $\beta_1(t_c)\cap \beta(x)\subseteq Y_t$ then we assign $\gamma(t_c)=y$ for some $y\in \tau^{-1}$. In this case, property $4$ directly holds. Otherwise there is one vertex $x\in \tau^{-1}(t)$ such that $\beta_1(t_c)\cap \beta(x)$ is not a subset of $Y_t$, then we assign $\gamma(t_c)=x$. Here too, property $4$ holds.

To complete the proof, we show that for all but at most one vertex $x\in \tau^{-1}(t)$, $\beta_1(t_c)\cap \beta(x)\subseteq Y_t$. Since $\cT$ is a nice tree partition,  $T[V_{t_c}]$ is a subtree of $T$, where  $V_{t_c} = \bigcup_{x\in V(\cT_{t_c})}\tau^{-1}(x)$. Next $V_{t_c}$ does not contain $r_T$ because $t_c\neq r_\cT$ and $\tau(r_t)=r_\cT$. Further $V_{t_c}$ is tree in the forest $T\setminus \tau^{-1}(t)$. Every vertex in $\tau^{-1}(t)$ has at most one neighbor in $V_{t_c}$ otherwise it will form a cycle. Since $T$ is a rooted tree there is at most one node $x\in \tau^{-1}(t)$ whose neighbor in $V_{t_c}$ is not an ancester(or parent) of $x$. Thus for all others $\beta(x)\cap \beta_1(t_c)\subseteq \sigma(x,\parent(x))\subseteq Y_T$. 
\end{proof}

Let $\gamma:V(\cT)\rightarrow V(T)$ and $Y_t\subseteq \beta_1(t)$ for each node $t\in \cT$ be function and sets given by Lemma~\ref{lemma:tree_partition_decomp_prop}. We now define a pair $(T^*,\beta^*)$ based on $T,\cT,\gamma$ and $Y_t$ that we will prove to be a tree decomposition of $G$ having all our desired properties including unbreakable bags. 
Let $T^*$ be a graph with $V(T^*)=V(T)\cup V(\cT)$ and $E(T^*)=\{(\tau(x),x):x\in T\}\cup\{(\gamma(t),t):t\in \cT\setminus \{r_{\cT}\}\}$. Also let $\beta^*:V(T^*)\rightarrow 2^{V(G)}$ be a function with $\beta^*(t)=Y_t$, for $t\in \cT$ and $\beta^*(x)=Y_{\tau (x)}\cup \beta(x)$ for $x\in T$. We now show that $(T^*,\beta^*)$ is a tree decomposition of $G$ (see \Cref{fig:final_td}).

\begin{lemma}\label{lemma:final_td}
$(T^*,\beta^*)$ is a rooted tree decomposition of $G$. Further $(T^*,\beta^*)$ satisfies the following properties:
\begin{enumerate}
    \item every adhesion of $(T^*,\beta^*)$ is of size at most $8k$
    \item every bag of $(T^*,\beta^*)$ is $(q+8k,k)$-unbreakable in $G$.
    \item $T^*$ has depth at most $2\lceil\log_2|V(G)|\rceil$
\end{enumerate}
\end{lemma}
\begin{proof}
 We first prove {\em $T^*$ is a tree}. By definition, $T^*$ has at most $|V(T^*)|-1$ edges. To show $T^*$ is connected, we prove that each vertex in $V(T^*)$ can reach $r_{\cT}$ in $T^*$. If $v\in \cT\setminus \{r_{\cT}\}$, then there is a path $(t_1=v),t_2,\cdots,(t_l=r_{\cT})$ from $v$ to $r_{\cT}$ in $\cT$. Let $i\in [l-1]$. Observe that $t_i$ is a child of $t_{i+1}$ in $\cT$ and $t_{i+1}=\tau(\gamma(t_i))$ by property $4$ in Lemma~\ref{lemma:tree_partition_decomp_prop}. Thus by construction of $T^*$, $(t_{i},\gamma(t_i))$ and $(\gamma(t_i),t_{i+1})\in E(T^*)$, where $t_{i+1}=\tau(\gamma(t_i))$. Further $\gamma(t_i)\neq \gamma(t_j)$ for each distinct pair $i,j\in [l-1]$ since $\tau(\gamma(t_i))=t_{i+1}\neq t_{j+1}=\tau(\gamma(t_j))$.   
 Thus $(t_1=v),\gamma(t_1),t_2,\gamma(t_2),\cdots,(t_l=r_\cT)$ is a path from $v$ to $r_{\cT}$ in $T^*$. Further notice that this path does not contain any vertex from $\tau^{-1}(v)$ since $\gamma(t_i)\notin \tau^{-1}(v)$ for $i\in [l-1]$. 
 For each $x\in T$, either $\tau(x)=r_{\cT}$ or there is a path from $\tau(x)$ to $r_{\cT}$ without using $x$ in $\cT$. Also by construction $(x,\tau(x))$ is an edge in $T^*$ and thus there is a path from $x$ to $r_{\cT}$.  

 Next we show that $(T^*,\beta^*)$ is a tree decomposition of $G$. For each vertex $x\in V(T)$, it holds that $x\in V(T^*)$ and $\beta(x)\subseteq \beta^*(x)$. Thus since $(T,\beta)$ is a tree decomposition, each vertex $v\in V(G)$ and edge $e\in E(G)$ is contained in some bag $\beta(x)$ in $(T,\beta)$ and thus in the bag $\beta^*(x)$ in $(T^*,\beta^*)$. Lastly we need to show that for each node $v\in V(G)$, the bags containing $v$ in $(T^*,\beta^*)$ form a connected component in $T^*$.
 %
 %
 We first make an observation relating $(\cT,\beta_1)$ and $(\cT^*,\beta^*)$. For $t\in \cT$, let $X_t=\{t\}\cup \tau^{-1}(t)$. Observe that for each $t\in \cT$, contracting the bags corresponding to vertices in $X_t$ in $(T^*,\beta^*)$ into a single bag yields us the tree decomposition $(\cT,\beta_1)$. 
 
 Let $v\in V(G)$ and let $t\in \cT$ be the node closest to the root in $\cT$ whose bag contains $v$, i.e $v\in \beta_1(t)$. Such a node exists since $(\cT,\beta_1)$ is a rooted tree decomposition. 
 Further let $\beta^*(l)$ and $\beta^*(l')$ be two bags containing $v$ in $(T^*,\beta^*)$, $l,l'\in X_t$. We show that each bag in the path from $l'$ to $l$ in $(T^*,\beta^*)$ contains $v$. 
 If $l=t,l'\in \tau^{-1}(t)$, then $(l,l')\in E(T^*)$. Otherwise $l,l'\in \tau^{-1}(t)$. Here $\beta(l)\cap \beta(l')\subseteq Y_t$ (Lemma~\ref{lemma:tree_partition_decomp_prop}, property 3) and thus $v$ belongs to $\beta^*(t)=Y_t$. Further $(l,t), (t,l')\in E(T^*)$ and thus $(l,t,l')$ is a path in $T^*$ whose each bag contains $v$. 

 From now, we fix $l$ to be a node defined as follows. If $v \in \beta^*(t)$, then $l=  t$; otherwise let $l \in \tau^{-1}(t)$ be the unique node in $X_t$ such that $v\in \beta^*(x)$. Let $l'$ be a node in $\cT\setminus X_t$ such that $v\in \beta^*(l')$. We show that each bag in the path from $l'$ to $l$ contains $v$. This will show that the bags containing $v$ in $(T^*,\beta^*)$ form a connected component in $T^*$. This is because all nodes can reach $l$ using a path whose each bag contains $v$.

Observe that $l'$ is an ancestor of $l$ and that $l'\in X_{t'}$, $l\in X_{t}$, where $t'\neq t$. Further let $P$ be the path from $l'$ to $l$ in $\cT$. We first show that all bags $\beta^*(y)$ where $y\in P\setminus X_{t}$ contain $v$. Then we divide further into subcases to show that the same is true for each $y\in P\cap X_t$ as well. 
Since $v\in \beta_1(t')$ and $v\in \beta_1(t)$, there is a path $(t_1=t'),t_2,\cdots,(t_q=t)$ from $t'$ to $t$ in $(\cT,\beta_1)$ in which each bag $\beta_1(t_i)$, $i\in [q-1]$ contains $v$. Further for each $i\in [q-1]$, $v\in Y_{t_i}$ since the adhesion $\beta_1(t_i)\cap \beta_1(t_{i+1})\subseteq Y_{t_i}$ by property 2 of Lemma~\ref{lemma:tree_partition_decomp_prop}.

The path $P\setminus X_t$ is either $(l',t_1,\gamma(t_1),t_2,\gamma(t_2),\cdots,t_{q-1})$ or $((t_1=l'),\gamma(t_1),t_2,\gamma(t_2),\cdots,t_{q-1})$. By definition of $\beta^*$, $\beta^*(t_j)=Y_j, j\in[q-1]$ and $\beta^*(\gamma(t_i))=Y_{\tau(\gamma(t_i))}\cup \beta(t_i)$, $i\in [q-2]$. By property 4 of Lemma~\ref{lemma:tree_partition_decomp_prop}, $\tau(\gamma(t_i))=t_{i+1}$ and so $Y_{t_{i+1}}\subseteq \beta^*(\gamma(t_i))$. 
Thus for each $i\in [q-1]$ and for each $j\in[q-2]$, $v\in Y_{t_i}$, $\beta^*(t_i)=Y_{t_i}$, and $Y_{t_{j+1}}\subseteq \beta^*(\gamma(t_j))$. This shows that for each $y\in P\setminus X_{t}$, $v\in \beta^*(y)$. 

Now we have cases for proving $v\in \beta^*(y)$ for each $y\in P\cap X_t$. By our assumption, $v\in \beta^*(l)$. (a) If $P\cap X_t=\{l\}$, then we are done. (b) If $|P\cap X_t|=2$ then $P\cap X_t$ has to be $(\gamma(t_{q-1}),t_q=l)$. Here by an argument similar to we did for $\gamma(t_j)$, $j\in [q-2]$, $Y_{t_{q}}\subseteq \beta^*(\gamma(t_{q-1}))$ and thus since $\beta^*(l)=Y_{t_{q}}$, $v\in \beta^*(\gamma(t_{q-1}))$. 
(c) Finally the only other possibility for $P\cap X_t$ is $(\gamma(t_{q-1}),t_q,l)$. Here $l\neq \gamma(t_{q-1})$. By property $4$ of Lemma~\ref{lemma:tree_partition_decomp_prop}, $\beta_1(t_{q-1})\cap \beta_1(t_q)\subseteq Y_{t_q}\cup \beta(\gamma(t_{q-1}))$. We know that $v\in \beta_1(t_{q-1})\cap \beta_1(t_q)$. Thus $v\in Y_{t_q}\cup \beta(\gamma(t_{q-1}))$. Since $\beta^*(\gamma(t_q))=Y_{t_q}\cup \beta(\gamma(t_{q-1}))$, $v\in \beta^*(\gamma(t_q))$. But here since $l\neq t$, $l$ should have been the unique node in $X_t$ whose bag contains $v$. So this case cannot occur.  This completes the proof of the connectivity property and show that $(T^*,\beta^*)$ is a tree decomposition.

We now {\em prove properties $1-3$}. $(1)$ By Lemma~\ref{lemma:tree_partition_decomp_prop}, for each $y\in \cT$, $|Y_t|\leq 8k$. Further for each $t\in \cT$, $\beta^*(t)=Y_t$ and for each $x\in \cT$, $\beta^*(x)\cap \beta^*(\parent(x))=\beta^*(x)\cap \beta^*(\tau(x))=(\beta(x)\cup Y_{\tau(x)}) \cap Y_{\tau(x)} = Y_{\tau(x)}$. Thus each adhesion in $(\cT^*,\beta^*)$ has size at most $8k$.
$(2)$ For each $t\in \cT$, $\beta^*(t)=Y_t$ and $|Y_t|\leq 8k$. For each $x\in T$, $\beta(x)$ is $(q,k)$-unbreakable and $\beta^*(x)=\beta(x)\cup Y_{\tau(x)}$. Thus each bag in $(T^*,\beta^*)$ is $(q+8k,k)$-unbreakable.
$(3)$ For each $x\in T$, $\parent(x)=\tau(X)$. Further for each $t\in \cT$, contracting the bag corresponding to vertices in $X_t$ in $(T^*,\beta^*)$ into a single bag yields $(\cT,\beta_1)$. Since the depth of $\cT$ is at most $\lceil\log_2|V(G)|\rceil$, the depth of $T^*$ is at most $2\lceil\log_2|V(G)|\rceil$.
\end{proof}
We are now ready to prove our main theorem.
\begin{theorem}\label{theorem:low_depth_ub_td_poly}
There exists a polynomial-time algorithm that takes input an $n$-vertex graph $G$ and positive integers $k$ and $q$, and a rooted tree decomposition $(T,\beta)$ of $G$ satisfying the following properties:
\begin{enumerate}
    \item every adhesion of $(T,\beta)$ is of size at most $k$
    \item every bag of $(T,\beta)$ is $(q,k)$-unbreakable in $G$
\end{enumerate}
and finds a compact tree decomposition $(T',\beta')$ of $G$ satisfying the following properties:
\begin{enumerate}
    \item every adhesion of $(T',\beta')$ is of size at most $8k$
    \item every bag of $(T',\beta')$ is $(q+8k,k)$-unbreakable in $G$.
    \item $T'$ has depth at most $2\lceil\log_2n\rceil$.
\end{enumerate}
\end{theorem}
\begin{proof}
Let $(T,\beta)$ be the input tree decomposition of $G$. We first compute a nice tree partition $(\cT,\tau)$ of $T$ using Lemma~\ref{lemma:tp_tree}. Then we obtain the tree decomposition $(\cT,\beta_1)$ of $G$ where $\beta_1:V(G)\rightarrow V(\cT)$ and $\beta_1(t)=\bigcup_{x\in \tau^{-1}(t)} \beta(x)$ -- it is a tree decomposition by Lemma~\ref{lemma:treedecomp_from_treepartition}.


Let $\beta^*:V(T^*)\rightarrow 2^{V(G)}$ be a function with $\beta^*(t)=Y_t$, for $t\in \cT$ and $\beta^*(x)=Y_{\tau (x)}\cup \beta(x)$ for $x\in T$. We compute the tree decomposition $(T^*,\beta^*)$ with $V(T^*)=V(T)\cup V(\cT)$ and $E(T^*)=\{(\tau(x),x):x\in T\}\cup\{(\gamma(t),t):t\in \cT\setminus \{r_{\cT}\}\}$. By Lemma~\ref{lemma:final_td} it satisfies all our required properties except compactness. We can in polynomial time obtain a compact tree decomposition $(T',\beta')$ whose each bag is a subset of some bag of $(T^*,\beta^*)$ and whose height is the same as $\cT^*$~\cite{BojanczykP16}. Thus the tree decomposition $(T',\beta')$ will satisfy all our required properties.
\end{proof}

\begin{proposition}[\cite{CyganKLPPSW21}]\label{prop:unbreakable-treedecomp}
Given an $n$-vertex graph $G$ and an integer $k$, one can in time $2^{\cO(k\log k)}n^{\cO(1)}$ compute a rooted compact tree decomposition $(T,\beta)$ of $G$ such that:
\begin{enumerate}
    \item Every adhesion of $(T,\beta)$ is of size at most $k$.
    \item Every bag $B_t$ of $(T,\beta)$ is $(i,i)$-unbreakable in $G$ for every $i\in [k]$. 
\end{enumerate}
\end{proposition}

The following corollary directly follows from a known result (\cite{CyganKLPPSW21}) that outputs a tree decomposition of a graph satisfying the premise of \Cref{theorem:lowdepth-near-unbreakable-treedecomp} in time $2^{\cO(k \log k)} \cdot n^{\cO(1)}$
Proposition~\ref{prop:unbreakable-treedecomp} and Theorem~\ref{theorem:low_depth_ub_td_poly}.

\begin{corollary}\label{theorem:lowdepth-near-unbreakable-treedecomp}
Given an $n$-vertex  graph $G$ and  an  integer $k$,  one  can  in  time $2^{\cO(k\log k)}n^{\cO(1)}$ compute a rooted compact tree decomposition $(T,\beta)$ of $G$ such that:
\begin{enumerate}
    \item Every adhesion of $(T,\beta)$ is of size at most $8k$.
    \item Every bag of $(T,\beta)$ is $(9k,k)$-unbreakable in $G$.
    \item $T$ has depth at most $2\lceil\log_2 n\rceil$.
\end{enumerate}
\end{corollary}

\section{Exact and Approximation algorithms} \label{sec:algo}
Let $(G,c,k,\sr,\chi)$ be an instance of \fb and let $n = |V(G)|$. 
We start by invoking the algorithm of Theorem~\ref{theorem:lowdepth-near-unbreakable-treedecomp} with $G$ and $k$ to obtain a rooted compact tree decomposition $(T, \beta)$ of $G$, having $(9k, k)$-edge-unbreakable bags and adhesions of size at most $8k$. This takes time $2^{\cO(k \log k)}n^{\cO(1)}$. Recall that an edge cut $(A, B)$ is $(\epsilon, \sr)$-{\em fair} if $A$ contains no more than $r_i(1+\epsilon)$  vertices colored $i$ and $B$ contains no more than $(c_i-r_i)(1+\epsilon)$ vertices colored $i$.
\begin{theorem} \label{thm:mainapprox}
Given an instance $(G,c,k,\sr,\chi)$ of \fb and $\epsilon>0$ there exists an algorithm that in time $2^{\cO(k \log k)}\cdot \lr{\frac{c}{\epsilon}}^{\cO(c)} \cdot n^{\cO(1)}$ finds an $(\epsilon, \sr)$-{\em fair} edge cut of $G$ if one exists and else returns no.
\end{theorem}
Given a subset $S\subseteq V(G)$, we use $\chi^\circ(S)$ to denote the $c$ length tuple where the $i^{th}$ entry is the number of vertices $v$ in $S$ having color $i$, i.e. $\chi(v)=i$. We remark that throughout this section, we use $^\circ$ to denote tuples of integers of length $c$. Further we use operators such as $+, -$, scalar multiplication, and $\lceil \rceil$ on tuples, which perform the respective operations on each entry in the tuple(s). 

For a node $t\in V(T)$ recall that $\gamma(t) = \bigcup_{s:\text{ descendant of }t} \beta(s)$, $\alpha(t) = \gamma(t) \backslash \sigma(t), \text{    }G_t = G[\gamma(t)] - E(G[\sigma(t)]).$ We perform bottom-up dynamic programming on $(T, \beta)$. 
For each node $t\in V(T)$, we first define a Boolean function  $f_t:\{0,\cdots, k\}\times 2^{\sigma(t)} \times \{0,\cdots,n\}^{c}\times \{0,\cdots,n\}^{c}\rightarrow \{ \textsf{True, False} \}$. For each integer $w\in \{0,\cdots, k\}$, subset $A_{t} \subseteq \sigma(t)$, and $c$ length tuples $a^{\circ}$ and $b^{\circ}$ with $a^{\circ},b^{\circ}\in \{0,\cdots, n\}^{c}$ we define  
\begin{definition} $f_t(w, A_t,a^{\circ},b^{\circ})=\mathsf{True}$ if there exists an edge cut $(A, B)$ of $G_t$ that satisfies the following properties:  
\begin{enumerate}
    \item $(A,B)$ has order at most $w$.
    \item $A\cap \sigma(t) = A_{t}$
    \item $\chi^{\circ}(A\cap \alpha(t))=a^\circ$
    \item $\chi^{\circ}(B\cap \alpha(t))=b^\circ$
\end{enumerate}
If such a cut does not exist, $f_t(w, A_{t},a^\circ,b^\circ)=\mathsf{False}$.
Further if $f_t(w, A_{t},a^\circ,b^\circ)=\mathsf{True}$, we say an edge cut $(A,B)$ of $G_t$ that satisfies properties $1-4$ \textbf{realizes} $f_t(w, A_{t},a^\circ,b^\circ)$.
\end{definition}
From the definition of $f_t$ one can make the following observation. Let $r$ be the root of $T$.
\begin{observation}
$(G,c^\circ,k,r^\circ,\chi)$ is a yes-instance to \fb if and only if for $f_{r}(k,\emptyset, r^\circ,c^\circ-r^\circ)= \mathsf{True}$, where $r$ is the root of $T$.
\end{observation}
 
In order to reduce the size of the domain of $f$ (and hence the running time), we work with the \emph{reduced} domain $D = \left\{(1+\delta)^i : i \ge 0\right\}$. This will approximate the number of vertices of each color at either side of the cut to the nearest power of $1+\delta$, where $\delta > 0$ is a parameter whose value will be fixed later.

Let $\CC_t$ be the set of all possible edge-cuts $(A,B)$ of $G_t$. To compute $f_t$ we have a {\em table} $M_t:\{0,\cdots, k\}\times 2^{\sigma(t)} \times D^c \times D^c\rightarrow \{ \CC_t\cup \bot\}$ that satisfies properties $M_t\rightarrow f_t$ and $f_t\rightarrow M_t$ (defined below in Definition \ref{defn:Mtof_error} and \ref{defn:ftoM_error}). $M_t$ will help us to approximately obtain $f_t$.
%
%
Let $z \ge 0$ be a sufficiently large constant; for example, $z=10$ suffices. We have Definition~\ref{defn:Mtof_error} and Definition~\ref{defn:ftoM_error} that will be crucial towards proving the correctness of the approximation algorithm.

\begin{definition}[Property $M_t\rightarrow f_t$]\label{defn:Mtof_error}
If $M_t(w,A_t,a^{\circ},b^{\circ})\neq \bot$ then $\exists$  $x^{\circ}\in \{0,\cdots,n\}$ and $y^{\circ}\in \{0,\cdots,n\}$ such that:
\begin{itemize}
    \item $f_t(w,A_t,x^{\circ},y^{\circ})=\true$ and $M_t(w,A_t,a^{\circ},b^{\circ})$ is an edge-cut that realizes $f_t(w,A_t,x^{\circ},y^{\circ})$
    \item $a^\circ\leq x^\circ \leq (1+\delta)^{z \cdot \hgt(t)\log^2n} \cdot a^{\circ}$
    \item $b^\circ\leq y^\circ \leq (1+\delta)^{z \cdot \hgt(t)\log^2n} \cdot b^{\circ}$
\end{itemize}
\end{definition}

\begin{definition}[Global-feasible edge cut]
 An edge cut $(A,B)$ is global-feasible if there exists an edge cut $(A',B')$ of $G$ having order at most $k$ which induces the cut $(A,B)$ on $A\cup B$. 
\end{definition}

\begin{definition}[Property $f_t\rightarrow M_t$]\label{defn:ftoM_error}
If $f_t(w,A_t,x^{\circ},y^{\circ})=\true$ and there is a global-feasible edge cut $(A,B)$ of $G_t$ that realizes it then $\exists$ $a^{\circ}\in D^c$ and $b^{\circ}\in D^c$ such that:
\begin{itemize}
    \item $M_t(w,A_t,a^{\circ},b^{\circ})\neq \bot$ 
     \item $a^\circ\leq x^\circ \leq (1+\delta)^{z \cdot \hgt(t)\log^2n} \cdot a^{\circ}$
    \item $b^\circ\leq y^\circ \leq (1+\delta)^{z \cdot \hgt(t)\log^2n} \cdot b^{\circ}$
\end{itemize}
\end{definition}

\begin{definition}[Good $M_t$]\label{defn:goodM}
For $t\in V(T)$, we say $M_t$ is {\em good} if it satisfies properties $M_t\rightarrow f_t$ and $f_t\rightarrow M_t$.
\end{definition}

\begin{lemma}\label{lemma:proof_root}
For each $\epsilon>0$ and $\delta=\frac{\epsilon}{2z\log^3 n}$, if $f_r(k,\phi,r^\circ,c^\circ-r^\circ)=\true$ and $M_r$ is good, then $\exists$ $a^\circ,b^\circ\in D^c$ such that $M_r(k,\phi,a^\circ,b^\circ)$ is a $(\epsilon,r^\circ)$-fair edge cut of $G$.  Here $r$ is the root of $T$.
\end{lemma}
\begin{proof} We first show a useful bound relating $\delta$ and $\epsilon$.
\begin{claim} \label{cl:deltaepsilon}
If $\delta \coloneqq \frac{\epsilon}{2z\log^3 n}$, then $(1+\delta)^{z \cdot\log^3n} \leq 1+\epsilon$. Furthermore, $\log_{1+\delta}n = \lr{\log n/\epsilon}^{\cO(1)}$.
\end{claim}
\begin{proof}
Follows from the fact that $\ln(1+\epsilon) \ge \frac{\epsilon}{1+\epsilon} \ge \frac{\epsilon}{2}$, since $\epsilon \in (0, 1)$, which implies that $(1+\delta)^{z \cdot \log^3 n} \le \exp\lr{\frac{\epsilon}{2z \log^3 n} \cdot z \log^3 n} \le \exp\lr{\frac{\ln(1+\epsilon)}{z \log^3 n} \cdot z \log^3 n} = 1+\epsilon$.
\end{proof}
Let $f_r(k,\phi,r^\circ,c^\circ-r^\circ)=\true$ and $M_r$ satisfy properties $M_r\rightarrow f_r$ and $f_r\rightarrow M_r$. Since $ht(T)=\log n$, by the previous claim $(1+\delta)^{z \cdot \hgt(t)\log^2n} \leq 1+\epsilon$. So by property $f_r\rightarrow M_r$, $\exists$ $a^\circ,b^\circ \in D^c$ such that:
\begin{itemize}
    \item $M_r(k,\phi,a^{\circ},b^{\circ})\neq \bot$ 
     \item $a^\circ\leq r^\circ \leq (1+\epsilon) \cdot a^{\circ}$
    \item $b^\circ\leq c^\circ-r^\circ \leq (1+\epsilon)\cdot b^{\circ}$
\end{itemize}
Further by property $M_r\rightarrow f_r$, since $M_r(k,\phi,a^{\circ},b^{\circ})\neq \bot$, $\exists$ $x^\circ,y^\circ\in \{0,\cdots,n\}$ such that:
\begin{itemize}
    \item $f_r(k,\phi,x^\circ,y^\circ)=\true$ and $M_r(k,\phi,a^{\circ},b^{\circ})$ is an edge-cut that realizes $f_r(k,\phi,x^\circ,y^\circ)$
     \item $a^\circ\leq x^\circ \leq (1+\epsilon) \cdot a^{\circ} \leq (1+\epsilon) \cdot r^{\circ}$
    \item $b^\circ\leq y^\circ \leq (1+\epsilon)\cdot b^{\circ} \leq (1+\epsilon)\cdot (r^{\circ}-c^\circ)$
\end{itemize}
Thus, $M_r(k,\phi,a^{\circ},b^{\circ})$ is a $(\epsilon,r^\circ)$-fair edge-cut of $G$. This completes our proof.
\end{proof}

Lemma~\ref{lemma:proof_root} shows us that computing a good table $M$ efficiently is sufficient for obtaining our final approximation. We now state as a theorem that we can compute a good $M_t$ assuming a good $M_{t'}$ has been computed for each $t'\in \child(t)$.
\begin{lemma}\label{lemma:alg_M_t}
There exists an algorithm that takes as input $t\in V(T)$, $\delta>0$, $(T,\beta)$, and a good $M_{t'}$ for each $t'\in \child(t)$ and computes a good $M_{t}$ in time $2^{\cO(k\log k)}(\log _{1+\delta}n)^{\cO(c)}n^{\cO(1)}$.  
\end{lemma}

We will prove Lemma~\ref{lemma:alg_M_t} later. Before that, we now prove our main result, Theorem~\ref{thm:mainapprox}, assuming Lemma~\ref{lemma:alg_M_t}.
\begin{proof}[Proof of Theorem~\ref{thm:mainapprox}]
Let $\delta \coloneqq \frac{\epsilon}{2z\log^3 n}$. In our algorithm we compute $M$ by computing good $M_t$ using Lemma~\ref{lemma:alg_M_t} for each $t\in V(T)$, bottom up, starting from leaves of $T$ to root of $T$. We finally go over each $a^\circ,b^\circ\in D^c$ and output a cut $M_r(k,\phi,a^\circ,b^\circ)$ that is a $(\epsilon,r^\circ)$-fair edge cut of $G$ if one exists.

The correctness follows directly from the definition of $f$ and Lemma~\ref{lemma:proof_root}.
The time taken by our algorithm is equal to the size of domain of $M$ times the time taken to compute each entry in $M$. 
The size of the domain of $M$ is at most $2^{\cO(k)}(\log _{1+\delta}n)^{\cO(c)}n^{\cO(1)}$ because $|D|\leq \log _{1+\delta}n$ and all adhesions in $(T,\beta)$ have size at most $8k$.
The time taken to compute each entry in $M$ is $2^{\cO(k\log k)}(\log _{1+\delta}n)^{\cO(c)}n^{\cO(1)} $ by Lemma~\ref{lemma:alg_M_t}.
Thus, the total time taken is $2^{\cO(k\log k)}(\log _{1+\delta}n)^{\cO(c)}n^{\cO(1)}$, which is $2^{\cO(k \log k)} \lr{\frac{c}{\epsilon}}^{\cO(c)} \cdot n^{\cO(1)}$ by \Cref{cl:deltaepsilon} and a standard case analysis on whether $c \le \frac{\log n}{\log\log n}$. 
\end{proof}

\subsection{Computing $M_t$: Proof of Lemma~\ref{lemma:alg_M_t}} In this subsection we prove Lemma~\ref{lemma:alg_M_t}. 
In particular we design an algorithm that takes as input a graph $G$, the tree decomposition $(T,\beta)$ and a node $t$ of $T$, together with dynamic programming tables $M_{t'}$ for every child $t'$ of $t$, and outputs the appropriate dynamic programming table $M_{t}$ (which is good) for $t$.
This algorithm is an adaptation of a similar step performed by Cygan et al.~\cite{CyganLPPS19} in their algorithm for the {\sc Minimum Bisection} problem. 
The algorithm of Cygan et al.~\cite{CyganLPPS19} proceeds by a random coloring step, followed by a ``knapsack''-like dynamic programming algorithm. 
Our algorithm proceeds in a similar manner, but faces the following key difficulty:
in order to keep time and space bounded by $f(k,c,\epsilon)n^{O(1)}$ we can only store approximate values in the knapsack dynamic programming table (the table satisfies soundness and completeness properties similar to Definition~\ref{defn:goodM}).
Therefore, after computing each entry of the table (from previous entries) we need to perform a rounding step that introduces a $(1 + (\frac{\epsilon}{\log n})^{\cO(1)})$ multiplicative factor in the error bound. 
The standard way of solving {\sc Knapsack} involves considering each item in the input one by one, however this would lead to the rounding error possibly accumulating and getting out of hand. 
We overcome this by organizing the dynamic program in a complete binary tree. That is, split the items in two equal sized groups, compute dynamic programming tables for the two groups recursively, and combine the dynamic programming tables to the two halves to a dynamic programming for all the items. 
This ensures that the total error is upper bounded by a multiplicative factor of $(1 + (\frac{\epsilon}{\log n})^{\cO(1)})^{\cO(\log^3 n)} = 1+\cO(\epsilon)$.  

We now begin the formal exposition.
%
The way our algorithm will work is by defining a chain of ``true'' functions which we will compute approximately and use to compute the approximate value of the function higher up in the chain.
First we define a useful object. 
\begin{definition}[$H_t$, $H_t\setminus X$]
Let $H_t$ be the graph obtained from $G[\beta(t)]$ after making each adhesion $\sigma(t')$, $t'\in \textsf{child}(t)\cup \{t\}$ a clique. 
Further for $X\subseteq \beta(t)$, let $H_t\setminus X$ be the graph obtained by removing the vertices of $X$ from $H_t$. We say $H_t\setminus X$ is a \textbf{refinement} of an edge cut $(A,B)$ of $G_t$ if $X$ and each component in $H_t\setminus X$ is either a subset of $A$ or a subset of $B$. 
\end{definition}
We have the following lemma regarding $\{P_1,\cdots,P_p\}$, the set of connected components in $H_t\setminus X$, $X\subseteq \beta(t)$. The Lemma follows directly from the definition of $H_t$ and $H_t\setminus X$. We remark that here we view a connected component $P_\ell,$ $\ell\in [p]$ as a set of vertices. See Figure~\ref{fig:adhesions} for an example of how the objects in Lemma~\ref{lemma:splitter_partition_adhesions} look.  

\begin{figure}
    \centering
    \includegraphics[scale=0.5,page=3]{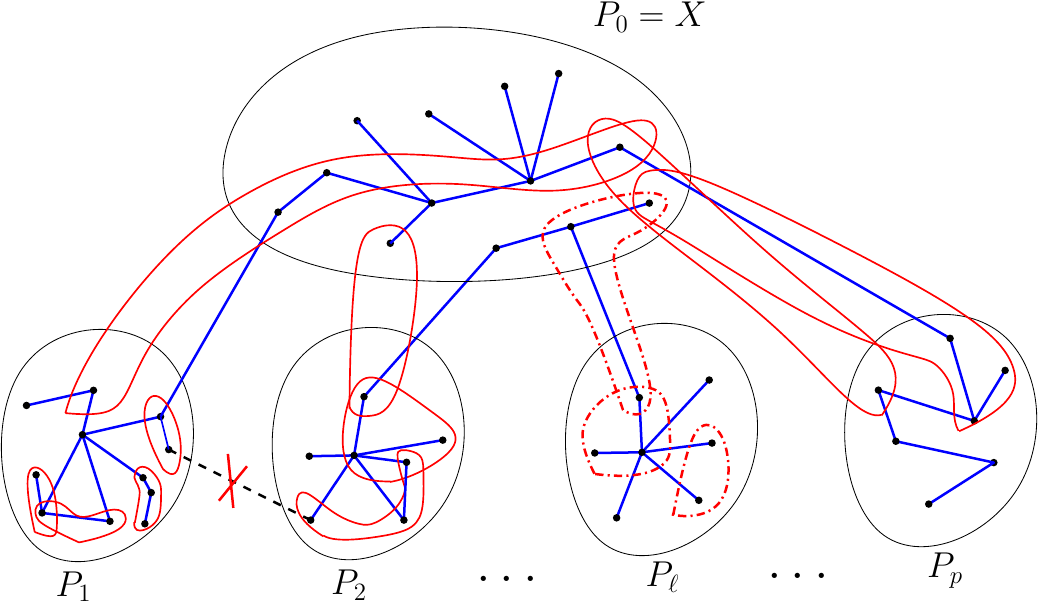}
    \caption{Exposition for \Cref{lemma:splitter_partition_adhesions}. The adhesions are shown in red color. The adhesions in $\mathcal{A}_{l}$ are shown in dashed red.}
    \label{fig:adhesions}
\end{figure}

\begin{lemma} \label{lemma:splitter_partition_adhesions}
 For $X\subseteq \beta(t)$, let $\{P_1,\cdots,P_p\}$ be the set of connected components in $H_t\setminus X$ and let $P_0=X$. Further for $\ell\in \{1,\cdots,p\}$, let $\mathcal{A}(\ell)=\{\sigma(t'):t'\in\textsf{child}(t)\cup\{t\},$ $\sigma(t')\subseteq P_\ell\cup P_0,\sigma(t')\cap P_\ell \neq \emptyset\}$ be the set of adhesions having vertices only from $P_\ell\cup P_0$. Then we have the following properties:
 \begin{itemize}
     \item For each $\ell_1,\ell_2\in \{1,\cdots,p\}$, $\ell_1\neq \ell_2$, there are no edges between $P_{\ell_1}$ and $P_{\ell_2}$ in $G$. 
     \item For each $\ell_1,\ell_2\in \{1,\cdots,p\}$, $\ell_1\neq \ell_2$, $\mathcal{A}(\ell_1)\cap \mathcal{A}(\ell_2)=\emptyset$. 
 \end{itemize}
\end{lemma}

For a subset $X\subseteq \beta(t)$, we define function $f^X_t$ and table $M^X_t$ to help us compute $M_t$.
For $X\subseteq \beta(t)$, we first define the Boolean function  $f^X_t:\{0,\cdots, k\}\times 2^{\sigma(t)} \times \{0,\cdots,n\}^{c}\times \{0,\cdots,n\}^{c}\rightarrow \{ \textsf{True, False} \}$. For each integer $w\in \{0,\cdots, k\}$, subset $A_{t} \subseteq \sigma(t)$, and $c$ length tuples $a^{\circ}$ and $b^{\circ}$ with $a^{\circ},b^{\circ}\in \{0,\cdots, n\}^{c}$ we define  
\begin{definition} $f^X_t(w, A_t,a^{\circ},b^{\circ})=\mathsf{True}$ if there exists an edge cut $(A, B)$ of $G_t$ that satisfies the following properties:  
\begin{enumerate}
    \item $(A,B)$ realizes $f_t(w,A_t,a^\circ,b^\circ)$ 
    \item $H_t\setminus X$ is a refinement of $(A,B)$
\end{enumerate}
If such a cut does not exist, $f^X_t(w, A_{t},a^\circ,b^\circ)=\mathsf{False}$.
Further if $f^X_t(w, A_{t},a^\circ,b^\circ)=\mathsf{True}$, we say an edge cut $(A,B)$ of $G_t$ that satisfies properties $1$ and $2$ \textbf{realizes} $f^X_t(w, A_{t},a^\circ,b^\circ)$.
\end{definition}
Recall that $\CC_t$ is the set of all possible edge-cuts $(A,B)$ of $G_t$. To compute $f^X_t$ we have a {\em table} $M^X_t:\{0,\cdots, k\}\times 2^{\sigma(t)} \times D^c \times D^c\rightarrow \{ \CC_t\cup \bot\}$ that satisfies properties $M^X_t\rightarrow f^X_t$ and $f^X_t\rightarrow M_t$ (defined below in Definition \ref{defn:XMtof_error} and \ref{defn:XftoM_error}).

\begin{definition}[Property $M^X_t\rightarrow f^X_t$]\label{defn:XMtof_error}
If $M^X_t(w,A_t,a^{\circ},b^{\circ})\neq \bot$ then $\exists$  $x^{\circ}\in \{0,\cdots,n\}$ and $y^{\circ}\in \{0,\cdots,n\}$ such that:
\begin{itemize}
    \item $f^X_t(w,A_t,x^{\circ},y^{\circ})=\true$ and $M^X_t(w,A_t,a^{\circ},b^{\circ})$ is an edge-cut that realizes $f^X_t(w,A_t,x^{\circ},y^{\circ})$
    \item $a^\circ\leq x^\circ \leq (1+\delta)^{z \cdot \hgt(t)\log^2n} \cdot a^{\circ}$
    \item $b^\circ\leq y^\circ \leq (1+\delta)^{z \cdot \hgt(t)\log^2n} \cdot b^{\circ}$
\end{itemize}
\end{definition}

\begin{definition}[Property $f^X_t\rightarrow M^X_t$]\label{defn:XftoM_error}
If $f^X_t(w,A_t,x^{\circ},y^{\circ})=\true$ and there is a global-feasible edge cut $(A,B)$ of $G_t$ that realizes it then $\exists$ $a^{\circ}\in D^c$ and $b^{\circ}\in D^c$ such that:
\begin{itemize}
    \item $M^X_t(w,A_t,a^{\circ},b^{\circ})\neq \bot$ 
     \item $a^\circ\leq x^\circ \leq (1+\delta)^{z \cdot \hgt(t)\log^2n} \cdot a^{\circ}$
    \item $b^\circ\leq y^\circ \leq (1+\delta)^{z \cdot \hgt(t)\log^2n} \cdot b^{\circ}$
\end{itemize}
\end{definition}
\begin{definition}[Good $M^X_t$]
For $X\subseteq \beta(t)$, we say $M^X_t$ is good if it satisfies properties $M^X_t\rightarrow f^X_t$ and $f^X_t\rightarrow M^X_t$.    
\end{definition}

As a first step of our algorithm we use Lemma~\ref{splittersetlemma} with $\beta(t)$, $s_1=9k$ and $s_2=8k^2+k$ to obtain a family $\mathcal{B}_t$ of subsets of $\beta(t)$. 
We have the following observation to capture the properties of $\mathcal{B}_t$ that directly follow from Lemma~\ref{splittersetlemma}.

\begin{observation}\label{obs:splitter_set_prop}
$\mathcal{B}_t$ has size at most $k^{\cO(k)}\cO(\log n)$ and each set $X\in \mathcal{B}_t$ is a subset of $\beta(t)$; $X\subseteq \beta(t)$. Further for any two disjoint subsets $X_1$ and $X_2$ of $\beta(t)$ of size at most $9k$ and $8k^2+k$, $\mathcal{B}_t$ contains a subset $X$ that satisfies $X_1\subseteq X$ and $X_2\cap X=\emptyset$.
\end{observation}

We have the following 
\begin{lemma}\label{lemma:f^X_t}
If $f_t(w,A_t,x^\circ,y^\circ)=\true$ and there is a global feasible edge cut $(A,B)$ of $G_t$ that realizes it, then $\exists X\in\BB(t)$ such that $(A,B)$ realizes $f^X_t(w,A_t,x^\circ,y^\circ)$.
\end{lemma}
\begin{proof}
$(A,B)$ is a global feasible edge cut of $G$ having weight at most $k$.  
Let $\mathcal{A}(t):= \{\sigma(s)\}\cup\{\sigma(t'):t'\text{ is a child of }t\}$ be the set of adhesions of children of $t$ and $\sigma(t)$. We say an adhesion $\sigma(t')$ in $\mathcal{A}(t)$ is broken if $\sigma(t')\cap A\neq \emptyset$ and $\sigma(t')\cap B\neq \emptyset$. 

Since $(T,\beta)$ is $(9k,k)$ unbreakable, either $|A\cap \beta(t)|$ or $|B\cap \beta(t)|$ is at most $9k$. Assume w.l.o.g that $|A\cap\beta(t)|\leq 9k$. Let $A^*=A\cap \beta(t)$.

Next let $B^*$ be the set of vertices in $B\cap \beta(t)$ either (i) incident to a cut edge in $(A,B)$ or (ii) part of a broken adhesion in $\mathcal{A}(t)$. Let $B^*=B'\cup B''$.
There are at most $k$ vertices in $B\cap \beta(t)$ adjacent to a cut edge in $(A,B)$ because $(A,B)$ has at most $k$ cut edges. Next there are at most $k$ broken adhesions because $T$ is compact; each broken adhesion can be associated with a unique cut edge. So there can be at most $8k^2$ vertices in a broken adhesion from $B\cap \beta(t)$ -- recall that each adhesion in $(T,\beta)$ has size at most $8k$. Thus $|B^*|\leq k+8k^2$.

By Observation~\ref{obs:splitter_set_prop}, $\exists X\in \BB_t$ such that $B^*\subseteq X$ and $A^*\cap X=\emptyset$. 
Lastly, it is easy to observe that $H_t\setminus X$ is a refinement of $(A,B)$ because $B^*\subseteq X$ and $A^*\cap X=\emptyset$. Thus $(A,B)$ also realizes $f^X_t(w,A_t,x^\circ,y^\circ)$, completing the proof.
\end{proof}
\begin{lemma}\label{lemma:alg_M^X_t}
There exists an algorithm that takes as input $t\in V(T)$, $X\subseteq \beta(t)$, $\delta>0$, $(T,\beta)$, and a good $M_{t'}$ for each $t'\in \child(t)$ and computes a good $M^X_{t}$ in time $2^{\cO(k)}(\log _{1+\delta}n)^{\cO(c)}n^{\cO(1)}$.  

\end{lemma}

 We now prove Lemma~\ref{lemma:alg_M_t} assuming Lemma~\ref{lemma:alg_M^X_t} which we prove in the next subsection.

 \begin{proof}[Proof of Lemma~\ref{lemma:alg_M_t}]
 We use Lemma~\ref{splittersetlemma} to compute $\BB_t$. Then for each $X\in \BB_t$ we use Lemma~\ref{lemma:alg_M^X_t} to compute $M^X_t$. Finally for each $(w,A_t,a^\circ,b^\circ)$ in the domain of $M_t$, we compute $M_t(w,A_t,a^\circ,b^\circ)$. We set $M_t(w,A_t,a^\circ,b^\circ)=\bot$ if $M^X_t(w,A_t,a^\circ,b^\circ)=\bot$ for all $X\in \BB_t$. Otherwise we set $M_t(w,A_t,a^\circ,b^\circ)$ to an arbitrary $M^X_t(w,A_t,a^\circ,b^\circ)\neq \bot$, $X\in \BB_t$. 

For correctness we need to show that $M_t$ is good. We first prove property $f_t\rightarrow M_t$.
Suppose $f_t(w,A_t,x^\circ,y^\circ)=\true$ and there is a global feasible edge cut $(A,B)$ of $G_t$ that realizes it. Then by Lemma~\ref{lemma:f^X_t}, $\exists X\in \BB_t$ such that $f^X_t(w,A_t,x^\circ,y^\circ)=\true$ and $(A,B)$ realizes it. By Lemma~\ref{lemma:alg_M^X_t}, $M^X_t$ is good and so $\exists a^\circ,b^\circ \in D^c$ such that 
\begin{itemize}
    \item $M^X_t(w,A_t,a^{\circ},b^{\circ})\neq \bot$ 
     \item $a^\circ\leq x^\circ \leq (1+\delta)^{z \cdot \hgt(t)\log^2n} \cdot a^{\circ}$
    \item $b^\circ\leq y^\circ \leq (1+\delta)^{z \cdot \hgt(t)\log^2n} \cdot b^{\circ}$
\end{itemize}
So by the working of our algorithm $M_t(w,A_t,a^\circ,b^\circ)\neq \bot$. This proves property $f_t\rightarrow M_t$.

Next we prove property $M_t\rightarrow f_t$. Suppose $M_t(w,A_t,a^\circ,b^\circ)\neq \bot$, then by the working of our algorithm $M_t(w,A_t,a^\circ,b^\circ)=M^X_t(w,A_t,a^\circ,b^\circ)$ for some $X\in \BB_t$. Now again since $M^X_t$ is good, $\exists x^\circ,y^\circ$ such that:
\begin{itemize}
    \item $f^X_t(w,A_t,x^{\circ},y^{\circ})=\true$ and $M^X_t(w,A_t,a^{\circ},b^{\circ})$ is an edge-cut that realizes $f^X_t(w,A_t,x^{\circ},y^{\circ})$
    \item $a^\circ\leq x^\circ \leq (1+\delta)^{z \cdot \hgt(t)\log^2n} \cdot a^{\circ}$
    \item $b^\circ\leq y^\circ \leq (1+\delta)^{z \cdot \hgt(t)\log^2n} \cdot b^{\circ}$
\end{itemize}
By definition of $f^X_t$ and $M^X_t$, this implies that $f_t(w,A_t,x^{\circ},y^{\circ})=\true$ and $M_t(w,A_t,a^{\circ},b^{\circ})=M^X_t(w,A_t,a^{\circ},b^{\circ})$ is an edge-cut that also realizes $f_t(w,A_t,x^{\circ},y^{\circ})$. This proves property $M_t\rightarrow f_t$.

We now compute the running time of our algorithm. Lemma~\ref{splittersetlemma} takes time $2^{\cO(k\log k)}\cO(\log n)$. Then for each $X\in \BB_t$, Lemma~\ref{lemma:alg_M^X_t} takes time $2^{\cO(k)}(\log _{1+\delta}n)^{\cO(c)}n^{\cO(1)}$ to compute $M^X_t$. Then computing $M_t$ takes time $(|\DD(M_t)|\cdot |\BB_t|)^{\cO(1)}$\footnote{we use $\DD(\cdot)$ to denote the domain of a function or table}.
 Thus the algorithm in total takes time $2^{\cO(k\log k)}(\log _{1+\delta}n)^{\cO(c)}n^{\cO(1)}$.
 \end{proof}

\subsection{Computing $M^X_t$: Proof of Lemma~\ref{lemma:alg_M^X_t}}
In this section we show how to compute a good $M^X_{t}$, $X\in \BB_t$ 
assuming a good $M_{t'}$ has been computed for each $t'\in \textsf{child}(t)$. 
To compute  $M^X_{t}$,
we will do a``knapsack style'' dynamic programming to fill tables (that we define soon) over a tree of bounded depth so that the error accumulation is $O(\log_2n)$ and not too much. We first define this tree.
\begin{definition}
For an integer $i\geq 0$, we denote by $\mathcal{T}^i$ a binary tree of depth at most $\lceil\log_2(i)\rceil$ having $V(T)=\{0,\cdots,i\}$ with $i$ being the root and $0$ being a leaf. 
\end{definition}


We note that we fix $X\in \BB_t$ throughout this section. Let $\{P_0=X,P_1,\cdots,P_p\}$ be the connected components in $H_t\setminus X$. We now define some notations to help define functions $g$ and $g_{\leq}$ and corresponding approximate tables $N$ and $N_{\leq}$ that will help us compute $M^X_{t}$.

Let $P_{\leq \ell} = \bigcup_{j\in V(\cT^p_\ell)} P_j$. Let $\cA(t)=\{\sigma(t'):t'\in \child(t)\}$. Given a non negative integer $\ell$, that is less than or equal to $p$, we define the set $\cA(\ell)$ to be
the set of all adhesions in the set $\cA(t)$ that only have vertices from $P_\ell \cup X$ and have a non empty intersection with $P_\ell$. 
$$\cA(\ell)=\{\sigma(t'):\sigma(t')\in \cA(t), \sigma(t')\subseteq P_\ell\cup X, \sigma(t')\cap P_l\neq \emptyset\}$$
Also, we denote by $\cA_{\leq}(\ell) = \bigcup_{\ell'\in V(\cT^p_\ell)} \cA(\ell')$, the union of all sets $\cA(\ell')$, where $\ell'$ is in the subtree $\cT^p_\ell$ (subtree of $\cT^p$ rooted at $\ell$). Further we define the graph $G(\ell) = G[X \cup P_\ell]\cup \bigcup_{\sigma(t')\in \cA(\ell)} G_{t'}$, be the subgraph induced by all vertices in $X \cup P_\ell\cup \bigcup_{\sigma(t')\in \cA(l)} V(G_{t'})$. Further we define the graph $G_{\leq}(\ell) = G[X \cup P_{\leq \ell}]\cup \bigcup_{\sigma(t')\in \cA_{\leq}(\ell)} G_{t'}$, be the subgraph induced by all vertices in $X \cup P_{\leq \ell}\cup \bigcup_{\sigma(t')\in \cA_{\leq}(\ell)} V(G_{t'})$.

We now define a function $g:\{0,\cdots,p\}\times\{0,\cdots,k\}\times 2^{\sigma(t)}\times \{0,\cdots,n\}^c\times \{0,\cdots,n\}^c \times \{\textsf{T},\textsf{F}\} \rightarrow \{\true,\false\}$.
For each integer $\ell\in \{0,\cdots,p\}$, $w\in \{0,\cdots,k\}$, $x_A\in \{\textsf{T},\textsf{F}\}$, $A_t\subseteq \sigma(t)$ and tuples $x^\circ,y^\circ \in \{0,\cdots,n\}^c$ we define:
\begin{definition} $g(\ell,w,A_t,x^\circ,y^\circ,x_A)=\true$ if there exists an edge cut $(A,B)$ of $G(\ell)-E(G[\sigma(t)])$ that satisfies the following properties:
\begin{enumerate}
    \item $(A,B)$ has order at most $w$.
    \item $A\cap \sigma(t) = A_t\cap V(G(\ell))$
    \item $P_\ell\subseteq A$ or $P_\ell\subseteq B$
    \item $\chi^\circ(A\setminus \sigma(t))=x^\circ$
    \item $\chi^\circ(B\setminus \sigma(t))=y^\circ$
    \item If $x_A=\textsf{T}$, then $X\subseteq A$ else $X\subseteq B$.
\end{enumerate}
If such a cut does not exist, $g(\ell,w,A_t,x^\circ,y^\circ,x_A)=\mathsf{False}$.
Further if $g(\ell,w,A_t,x^\circ,y^\circ,x_A)=\mathsf{True}$, we say an edge cut $(A,B)$ of $G(\ell)-E(G[\sigma(t)])$ that satisfies all the properties \textbf{realizes} $g(\ell,w,A_t,x^\circ,y^\circ,x_A)$.
\end{definition}


We define another function $g_{\leq}:\{0,\cdots,p\}\times\{0,\cdots,k\}\times 2^{\sigma(t)}\times \{0,\cdots,n\}^c\times \{0,\cdots,n\}^c \times \{\textsf{T},\textsf{F}\}\rightarrow \{\true,\false\}$.
For each integer $w\in \{0,\cdots,k\}$, $\ell\in \{0,\cdots,p\}$, $A_t\subseteq \sigma(t)$, $x_A\in \{\textsf{T},\textsf{F}\}$ and tuples $a^\circ,b^\circ \in \{0,\cdots,n\}^c$ we define: 
\begin{definition}
$g_{\leq}(\ell,w,A_t,x^\circ,y^\circ,x_A)=\true$ if there exists an edge cut $(A,B)$ of $G_{\leq}(\ell)-E(G[\sigma(t)])$ that satisfies the following properties:
\begin{itemize}
    \item $(A,B)$ has order at most $w$
    \item $A\cap \sigma(t) = A_t\cap V(G_{\leq}(\ell))$
    \item $\forall \ell'\in V(\cT^p_\ell)$, $P_{\ell'}\subseteq A$ or $P_{\ell'}\subseteq B$.
    \item $\chi^\circ(A\setminus \sigma(t))=x^\circ$
    \item $\chi^\circ(B\setminus \sigma(t))=y^\circ$
    \item If $x_A=\textsf{T}$, then $X\subseteq A$ and $X\subseteq B$ otherwise.
    \end{itemize}
    If such a cut does not exist, $g_{\leq}(\ell,w,A_t,x^\circ,y^\circ,x_A)=\mathsf{False}$.
Further if $g_{\leq}(\ell,w,A_t,x^\circ,y^\circ,x_A)=\mathsf{True}$, we say an edge cut $(A,B)$ of $G_{\leq}(\ell)-E(G[\sigma(t)])$ that satisfies all the properties \textbf{realizes} $g_{\leq}(\ell,w,A_t,x^\circ,y^\circ,x_A)$.
\end{definition}

In order to compute $g$ and $g_{\leq}$ approximately, we have {\em tables} $N$ and $N_{\leq}$. 
Let $C_{\ell}$ be the set of all possible edge cuts of $G(\ell)-E(G[\sigma(t)])$ and let $C_{\leq\ell}$ be the set of all possible edge cuts of $G_{\leq}(\ell)-E(G[\sigma(t)])$.
Further let $N:\{0,\cdots,p\}\times \{0,\cdots,k\} \times 2^{\sigma(t)}\times D^c \times D^c \times \{\textsf{T},\textsf{F}\} \rightarrow \{C_{\ell},\bot\}$ and $N_{\leq}:\{0,\cdots,p\}\times \{0,\cdots,k\} \times 2^{\sigma(t)} \times D^c \times D^c \times \{\textsf{T},\textsf{F}\} \rightarrow \{C_{\leq \ell},\bot\}$. We will compute $N_{\leq}$ such that it satisfies properties $N^\ell_{\leq}\rightarrow g^\ell_{\leq}$ and $g^\ell_{\leq}\rightarrow N^\ell_{\leq}$ (defined below) and use it to set $M^X_t$.
%

\begin{definition}[Property $N^\ell_{\leq}\rightarrow g^\ell_{\leq}$]\label{lemma:N_leqtog_leq}
If $N_{\leq}(\ell,w,A_t,a^{\circ},b^{\circ},x_A)\neq \bot$ then $\exists$ $x^{\circ},y^{\circ}\in \{0,\cdots,n\}^c$ such that:
\begin{itemize}
    \item $g_{\leq}(\ell,w,A_t,x^{\circ},y^{\circ},x_A)=\true$ and $N_{\leq}(\ell,w,A_t,a^\circ,b^\circ,x_A)$ realizes $g_{\leq}(\ell,w,A_t,x^{\circ},y^{\circ},x_A)$.
    \item $a^\circ\leq x^\circ \leq (1+\delta)^{z \cdot ((\hgt(t)-1)\log^2n+\hgt(\ell)\log n)} \cdot a^{\circ}$
    \item $b^\circ\leq y^\circ \leq (1+\delta)^{z \cdot ((\hgt(t)-1)\log^2n+\hgt(\ell)\log n)} \cdot b^{\circ}$
\end{itemize}
Here $\hgt(\ell)$ is the height of node $\ell$ in tree $\mathcal{T}^p$.
\end{definition}

\begin{definition}[Property $g^\ell_{\leq}\rightarrow N^\ell_{\leq}$]\label{lemma:g_leqtoN_leq}
If $g_{\leq}(\ell,w,A_t,x^{\circ},y^{\circ},x_A)=\true$ and there is a global-feasible edge cut $(A,B)$ of $G_{\leq}(\ell)-E(G[\sigma(t)])$ that realizes it then $\exists$ $a^{\circ},b^\circ\in D^c$ such that:
\begin{itemize}
    \item $N_{\leq}(\ell,w,A_t,a^{\circ},b^{\circ},x_A)\neq \bot$
     \item $a^\circ\leq x^\circ \leq (1+\delta)^{z \cdot ((\hgt(t)-1)\log^2n+\hgt(\ell)\log n)} \cdot a^{\circ}$
    \item $b^\circ\leq y^\circ \leq (1+\delta)^{z \cdot ((\hgt(t)-1)\log^2n+\hgt(\ell)\log n)} \cdot b^{\circ}$
\end{itemize}
\end{definition}

\begin{definition}[Good $N^\ell_{\leq}$]
We say $N^\ell_{\leq}$ is good if it satisfies properties $N^\ell_{\leq}\rightarrow g^\ell_{\leq}$ and $g^\ell_{\leq}\rightarrow N^\ell_{\leq}$.    
\end{definition}

\begin{definition}[Property $N\rightarrow g$]\label{lemma:Ntog}
If $N(\ell,w,A_t,a^{\circ},b^{\circ},x_A)\neq \bot$ then $\exists$ $x^{\circ},y^{\circ}\in \{0,\cdots,n\}^c$ such that:
\begin{itemize}
    \item $g(\ell,w,A_t,x^{\circ},y^{\circ},x_A)=\true$ and $N(\ell,w,A_t,a^\circ,b^\circ,x_A)$ realizes $g(\ell,w,A_t,x^{\circ},y^{\circ},x_A)$.
    \item $a^\circ\leq x^\circ \leq (1+\delta)^{z \cdot (\hgt(t)-1)\log^2n+2\log n} \cdot a^{\circ}$
    \item $b^\circ\leq y^\circ \leq (1+\delta)^{z \cdot (\hgt(t)-1)\log^2n+2\log n} \cdot b^{\circ}$
\end{itemize}
Here $\hgt(\ell)$ is the height of node $\ell$ in tree $\mathcal{T}^p$ and ht
\end{definition}

\begin{definition}[Property $g\rightarrow N$]\label{lemma:gtoN}
If $g(\ell,w,A_t,x^{\circ},y^{\circ},x_A)=\true$ and there is a global-feasible edge cut $(A,B)$ of $G(\ell)-E(G[\sigma(t)])$ that realizes it then $\exists$ $a^{\circ},b^\circ\in D^c$ such that:
\begin{itemize}
    \item $N(\ell,w,A_t,a^{\circ},b^{\circ},x_A)\neq \bot$
     \item $a^\circ\leq x^\circ \leq (1+\delta)^{z \cdot (\hgt(t)-1)\log^2n+2\log n} \cdot a^{\circ}$
    \item $b^\circ\leq y^\circ \leq (1+\delta)^{z \cdot (\hgt(t)-1)\log^2n+2\log n} \cdot b^{\circ}$
\end{itemize}
\end{definition}

\begin{definition}[Good $N$]
We say $N$ is good if it satisfies properties $N\rightarrow g$ and $g\rightarrow N$.    
\end{definition}

We now state that a good $N$ and $N_{\leq}$ can be computed efficiently. We will prove Lemma~\ref{lemma:N_leq} in the next subsection.
\begin{lemma}\label{lemma:N}
    There exists an algorithm that computes a good $N$ in time $2^{\cO(k)}(\log _{1+\delta}n)^{\cO(c)}n^{\cO(1)}$.
\end{lemma}
\begin{lemma}\label{lemma:N_leq}
   There exists an algorithm that takes as input a good $N_{\leq}(q,.)$ for each $q\in \child(l)$ and a good $N(\ell,.)$ and returns a good $N_{\leq}(\ell,.)$ in time $2^{\cO(k)}(\log _{1+\delta}n)^{\cO(c)}n^{\cO(1)}$.
\end{lemma}

We are now ready to prove Lemma~\ref{lemma:alg_M^X_t} that gives an algorithm for computing $M^X_t$ assuming Lemma~\ref{lemma:N_leq}.
\begin{proof}[Proof of Lemma~\ref{lemma:alg_M^X_t}]
We first compute a good $N_{\leq}$ using Lemma~\ref{lemma:N_leq}.  Let $(w,A_t,a^\circ,b^\circ)\in \DD(M^X_t)$. We now set $M^X_t(w,A_t,a^\circ,b^\circ)$ using $N_{\leq}$. Let $M_T = N_{\leq}(p,w,A_t,a^\circ,b^\circ,\textsf{T})$ and $M_F = N_{\leq}(p,w,A_t,a^\circ,b^\circ,\textsf{F})$.

\begin{align*}
M^X_t(w,A_t,a^\circ,b^\circ) &= \begin{cases}
    \bot & \text{, if } M_T=M_F=\bot \\
    M_F & \text{, if } M_F\neq \bot \\
    M_T & \text{, otherwise}
\end{cases}
\end{align*}


We now show that the $M^X_t$ we computed is good. We first show property $M^X_t\rightarrow f^X_t$. 
Let $M^X_t(w,A_t,a^\circ,b^\circ)\neq \bot $. Then it is set to $N_{\leq}(p,w,A_t, a^\circ,b^\circ,x_A=\textsf{T})$ or $N_{\leq}(p,w,A_t, a^\circ,b^\circ,x_A=\textsf{F})$. Since $N_{\leq}$ is good, $\exists$ $x^{\circ},y^{\circ}\in \{0,\cdots,n\}^c$ such that:
\begin{itemize}
    \item $g_{\leq}(p,w,A_t,x^{\circ},y^{\circ},x_A)=\true$ and $N_{\leq}(p,w,A_t,a^\circ,b^\circ,x_A)$ realizes $g_{\leq}(p,w,A_t,x^{\circ},y^{\circ},x_A)$.
    \item $a^\circ\leq x^\circ \leq (1+\delta)^{z \cdot ((\hgt(t)-1)\log^2n+\hgt(p)\log n)}\} \cdot a^{\circ}$
    \item $b^\circ\leq y^\circ \leq (1+\delta)^{z \cdot ((\hgt(t)-1)\log^2n+\hgt(p)\log n)}\} \cdot b^{\circ}$
\end{itemize}
By definition of $g_{\leq}$, $f^X_{t}(w,A_t,x^{\circ},y^{\circ})=\true$ and $M^X_t(w,A_t,a^\circ,b^\circ)$ realizes it.
Since $\hgt(p)\leq \log n$, $x^\circ\leq (1+\delta)^{z \cdot \hgt(t)\log^2n} \cdot a^{\circ}$ and $y^\circ\leq (1+\delta)^{z \cdot \hgt(t)\log^2n}\cdot b^{\circ}$.

We now show property $f^X_t\rightarrow M^X_t$. Suppose $f^X_t(w,A_t,x^{\circ},y^{\circ})=\true$ and there is a global-feasible cut $(A,B)$ that realizes it, then by definition of $f^X_t$ and $g_{\leq}$, $g_{\leq}(p,w,A_t,x^{\circ},y^{\circ},x_A)=\true$ for $x_A=\true$ or $x_A=\false$ and $(A,B)$ realizes it. Now since $N_{\leq}$ is good, then $\exists$ $a^{\circ},b^\circ\in D^c$ such that:
\begin{itemize}
    \item $N_{\leq}(p,w,A_t,a^{\circ},b^{\circ},x_A)\neq \bot$
     \item $a^\circ\leq x^\circ \leq (1+\delta)^{z \cdot ((\hgt(t)-1)\log^2n+\hgt(p)\log n))}\} \cdot a^{\circ} \leq (1+\delta)^{z \cdot \hgt(t)\log^2n} \cdot a^\circ$
    \item $b^\circ\leq y^\circ \leq (1+\delta)^{z \cdot ((\hgt(t)-1)\log^2n+\hgt(p)\log n))}\} \cdot b^{\circ} \leq (1+\delta)^{z \cdot \hgt(t)\log^2n} \cdot b^{\circ}$
\end{itemize}
By the working of our algorithm, $M^X_t(w,A_t,a^\circ,b^\circ)\neq \bot$. The running time follows from Lemma~\ref{lemma:N_leq} and the bound on size of domain of $M^X_t$. 
\end{proof}

We now show how to compute a good $N_{\leq}$ to prove Lemma~\ref{lemma:N_leq}. For this, we first define some necessary notations. 

\begin{definition}
For a tree $\cT$ and a node $x\in V(\cT)$, let $\faml_{\cT}(x) = \child_{\cT}(x) \cup \LR{x}$. A $w$-configuration is a triple $(\nu_w,\nu^\circ_{a},\nu^\circ_{b})$ where:
\begin{itemize}
    \item $\nu_w:\faml_{\cT}(l)\rightarrow [n]$ is a function such that $\displaystyle \sum_{q\in \faml_{\cT}(x)} \nu(q)\leq w$
    \item $\nu^\circ_{a}:\faml_{\cT}(l)\rightarrow D^c$   
    \item $\nu^\circ_{b}:\faml_{\cT}(l)\rightarrow D^c$ 
\end{itemize}
\end{definition}

\begin{definition}
A $w$-configuration $(\nu_w,\nu^\circ_{a},\nu^\circ_{b})$ is $(\ell,w,A_t,a^\circ,b^\circ,x_A)$-feasible, where $(\ell,w,A_t,a^\circ,b^\circ,x_A)\in \mathcal{D}(N_{\leq})$ if it satisfies:
\begin{itemize}
    \item If $x_A=\textsf{T}$:
    \begin{itemize}
        \item $\displaystyle \big\lceil \sum_{q\in \faml_{\cT}(\ell)} \nu^\circ_{a}(q) - |\child_{\cT}(\ell)|\cdot \chi^\circ(X\setminus \sigma(t)) \big\rceil\leq a^\circ$ 
        \item $\displaystyle \big\lceil \sum_{q\in \faml_{\cT}(\ell)} \nu^\circ_{b}(q)\big\rceil\leq b^\circ$ 
    \end{itemize}
    \item If $x_A=\textsf{F}$:
    \begin{itemize}
        \item $\displaystyle \big\lceil \sum_{q\in \faml_{\cT}(\ell)} \nu^\circ_{a}(q)\big\rceil\leq a^\circ$
        \item $\displaystyle \big\lceil \sum_{q\in \faml_{\cT}(\ell)} \nu^\circ_{b}(q) - |\child_{\cT}(\ell)|\cdot \chi^\circ(X\setminus \sigma(t)) \big\rceil\leq b^\circ$ 
    \end{itemize}
    \item $N(\ell,\nu_w(\ell),A_t,\nu^\circ_{a}(\ell),\nu^\circ_{b}(\ell),x_A)\neq \bot$
    \item For each $q\in \child(\ell)$, $N_{\leq}(q,\nu_w(q),A_t,\nu^\circ_{a}(q),\nu^\circ_{b}(q),x_A)\neq \bot$
\end{itemize}
\end{definition}
Next we compute $N_{\leq}$ using good $N$ and $N_{\leq}$ of its children to prove Lemma~\ref{lemma:N_leq}.
\begin{proof}[Proof of Lemma~\ref{lemma:N_leq}]
If there is no $(\ell,w,A_t,a^\circ,b^\circ,x_A)$-feasible $w$-configuration, then set $N_{\leq}(\ell,w,A_t,a^\circ,b^\circ,x_A) =\bot$. Otherwise let $(\nu_w,\nu^\circ_{a},\nu^\circ_{b})$ be an $(\ell,w,A_t,a^\circ,b^\circ,x_A)$-feasible $w$-configuration.

Let $\GG$ be the family of graphs $G(\ell)$ and $G_{\leq}(q)$, for each $q\in \child(\ell)$. Observe that for any two graph $G_1,G_2\in \GG$, the graphs $G_1\setminus X$ and $G_2\setminus X$ do not have any common vertices. Further there are no edges between $V(V(G_1)\setminus X)$ and $V(V(G_2)\setminus X)$. Thus we can obtain the cut $(A^*,B^*)$ of $G_\leq(\ell)$ by combining the $A$-side and $B$-side of the cuts $N_{\leq}(q,\nu_w(q),A_t,\nu^\circ_{a}(q),\nu^\circ_{b}(q),x_A)$ for each $q\in \child(\ell)$ and the cut $N(\ell,\nu_w(\ell),A_t,\nu^\circ_{a}(\ell),\nu^\circ_{b}(\ell),x_A)$. The common vertex set is $X$ and by definition of $N_{\leq}$ and $N$, $X\subseteq A^*$ if $x_A=\mathsf{T}$ or $X\subseteq B^*$ if $x_A=\mathsf{F}$.
In this case set $N_{\leq}(\ell,w,A_t,a^\circ,b^\circ,x_A)=(A^*,B^*)$.

On careful observation, one can see that $N_{\leq}$ is good for $\ell$ using the fact that $N$ is good and $N_{\leq}$ is good for each child $q$ of $\ell$ along with the notion of feasible $w$-configuration. 
\end{proof}




\subsection{Computing $N^\ell$: Proof of Lemma~\ref{lemma:N}}
We note that we fix $\ell\in p$ and the connected component $P_\ell$ throughout this section. We now define some notations to help define function $h_{\leq}$ and corresponding approximate table $H_{\leq}$ that will help us compute $N^\ell$.

Let the set of adhesions $\AA(\ell)$ associated with $P_\ell$ be $\mathcal{A}(\ell)=\{A_1,A_2,\cdots,A_z\}$, where $z=|\mathcal{A}(\ell)|$. We now perform a dynamic programming over the binary tree $\mathcal{T}^z$ for computing $N^\ell$. This will be very similar to what we did for computing $N_{\leq}$.

Let $\AA(l,m) = \bigcup_{j\in V(\cT^z_m)} A_j$. 
We define the graph $G(\ell,m) = G[X \cup P_{\ell}]\cup \bigcup_{\sigma(t')\in \cA(\ell,m)} G_{t'}$, to be the subgraph induced by all vertices in $X \cup P_{\ell}\cup \bigcup_{\sigma(t')\in \cA(\ell,m)} V(G_{t'})$.

We now define a function $h:\{0,\cdots,z\}\times\{0,\cdots,k\}\times 2^{\sigma(t)}\times \{0,\cdots,n\}^c\times \{0,\cdots,n\}^c \times \{\textsf{T},\textsf{F}\} \times \{\textsf{T},\textsf{F}\} \rightarrow \{\true,\false\}$.
For each integer $m\in \{0,\cdots,z\}$, $w\in \{0,\cdots,k\}$, $x_A\in \{\textsf{T},\textsf{F}\}$, $x_\ell\in \{\textsf{T},\textsf{F}\}$, $A_t\subseteq \sigma(t)$ and tuples $x^\circ,y^\circ \in \{0,\cdots,n\}^c$ we define:
\begin{definition} $h(m,w,A_t,x^\circ,y^\circ,x_A,x_\ell)=\true$ if there exists an edge cut $(A,B)$ of $G(\ell,m)-E(G[\sigma(t)])$ that satisfies the following properties:
\begin{enumerate}
    \item $(A,B)$ has order at most $w$.
    \item $A\cap \sigma(t) = A_t\cap V(G(\ell,m))$
    \item $\chi^\circ(A\setminus \sigma(t))=x^\circ$
    \item $\chi^\circ(B\setminus \sigma(t))=y^\circ$
    \item If $x_A=\textsf{T}$, then $X\subseteq A$ else $X\subseteq B$.
    \item If $x_\ell=\textsf{T}$, then $P_\ell\subseteq A$ else $P_\ell \subseteq B$
\end{enumerate}
If such a cut does not exist, $h(m,w,A_t,x^\circ,y^\circ,x_A,x_\ell)=\mathsf{False}$.
Further if $h(m,w,A_t,x^\circ,y^\circ,x_A,x_\ell)=\mathsf{True}$, we say an edge cut $(A,B)$ of $G(\ell,m)-E(G[\sigma(t)])$ that satisfies all the properties \textbf{realizes} $h(m,w,A_t,x^\circ,y^\circ,x_A,x_\ell)$.
\end{definition}

In order to compute $h$ approximately, we have table $H$. 
Let $C_{\ell,m}$ be the set of all possible edge cuts of $G(\ell,m)-E(G[\sigma(t)])$.
Further let $H:\{0,\cdots,z\}\times\{0,\cdots,k\}\times 2^{\sigma(t)}\times \{0,\cdots,n\}^c\times \{0,\cdots,n\}^c \times \{\textsf{T},\textsf{F}\} \times \{\textsf{T},\textsf{F}\} \rightarrow \{C_{\ell,m},\bot\}$. We will compute $H$ such that it satisfies properties $H\rightarrow h$ and $h\rightarrow H$ (defined below) and use it to set $N^\ell$.
%

\begin{definition}[Property $H\rightarrow h$]\label{lemma:Htoh}
If $H(m,w,A_t,a^{\circ},b^{\circ},x_A,x_\ell)\neq \bot$ then $\exists$ $x^{\circ},y^{\circ}\in \{0,\cdots,n\}^c$ such that:
\begin{itemize}
    \item $h(m,w,A_t,x^{\circ},y^{\circ},x_A,x_\ell)=\true$ and $H(m,w,A_t,a^\circ,b^\circ,x_A,x_{\ell})$ realizes $h(m,w,A_t,x^{\circ},y^{\circ},x_A,x_\ell)$.
    \item $a^\circ\leq x^\circ \leq (1+\delta)^{z \cdot (\hgt(t)-1)\log^2n+2\hgt(m)}\} \cdot a^{\circ}$
    \item $b^\circ\leq y^\circ \leq (1+\delta)^{z \cdot (\hgt(t)-1)\log^2n+2\hgt(m)}\} \cdot b^{\circ}$
\end{itemize}
Here $\hgt(m)$ is the height of node $m$ in tree $\mathcal{T}^z$.
\end{definition}

\begin{definition}[Property $h\rightarrow H$]\label{lemma:htoH}
If $h(m,w,A_t,x^{\circ},y^{\circ},x_A,x_{\ell})=\true$ and there is a global-feasible edge cut $(A,B)$ of $G(\ell,m)-E(G[\sigma(t)])$ that realizes it then $\exists$ $a^{\circ},b^\circ\in D^c$ such that:
\begin{itemize}
    \item $H(m,w,A_t,a^{\circ},b^{\circ},x_A,x_B)\neq \bot$
     \item $a^\circ\leq x^\circ \leq (1+\delta)^{z \cdot (\hgt(t)-1)\log^2n+2\hgt(m)}\} \cdot a^{\circ}$
    \item $b^\circ\leq y^\circ \leq (1+\delta)^{z \cdot (\hgt(t)-1)\log^2n+2\hgt(m)}\} \cdot b^{\circ}$
\end{itemize}
\end{definition}

\begin{definition}[Good $H$]
We say $H$ is good if it satisfies properties $H\rightarrow h$ and $h\rightarrow H$.    
\end{definition}

Given a good $H$, we can compute $N(\ell,.)$ as follows: $N(\ell,w,A_t,a^\circ,b^\circ,x_A)=\bot$ if $H(z,w,A_t,a^\circ,b^\circ,x_A,x_\ell)=\bot$ for both $x_\ell=\textsf{T}$ and $x_\ell=\textsf{F}$. Else we can set $N(\ell,w,A_t,a^\circ,b^\circ,x_A)$ to $H(z,w,A_t,a^\circ,b^\circ,x_A,x_\ell)$ such that $x_\ell\neq \bot$. We can easily verify $N$ is good if $H$ is good. This proves Lemma~\ref{lemma:N}.

A good $H$ can be computed by using $M^X_{t'}$, $t'\in \child(t)$. This can be done in a similar manner to $N_{\leq}$ by defining a feasible configuration and merging cuts.

\begin{definition}
A $w$-configuration $(\nu_w,\nu^\circ_{a},\nu^\circ_{b})$ is $(m,w,A_t,a^\circ,b^\circ,x_A,x_{\ell})$-feasible if it satisfies the following:
\begin{itemize}
    \item If $x_A=\textsf{T}$ and $x_\ell=\textsf{T}$:
    \begin{itemize}
        \item $\displaystyle \big\lceil \sum_{q\in \faml_{\cT}(\ell)} \nu^\circ_{a}(q) - |\child_{\cT}(\ell)|\cdot \chi^\circ(X\setminus \sigma(t)) - |\child_{\cT}(\ell)|\cdot \chi^\circ(P_{\ell}\setminus \sigma(t)) \big\rceil\leq a^\circ$ 
        \item $\displaystyle \big\lceil \sum_{q\in \faml_{\cT}(\ell)} \nu^\circ_{b}(q)\big\rceil\leq b^\circ$ 
    \end{itemize}
    \item If $x_A=\textsf{T}$ and $x_\ell=\textsf{F}$:
    \begin{itemize}
        \item $\displaystyle \big\lceil \sum_{q\in \faml_{\cT}(\ell)} \nu^\circ_{a}(q) - |\child_{\cT}(\ell)|\cdot \chi^\circ(X\setminus \sigma(t)) \big\rceil\leq a^\circ$ 
        \item $\displaystyle \big\lceil \sum_{q\in \faml_{\cT}(\ell)} \nu^\circ_{b}(q) - |\child_{\cT}(\ell)|\cdot \chi^\circ(P_{\ell}\setminus \sigma(t))\big\rceil\leq b^\circ$ 
    \end{itemize}
    \item If $x_A=\textsf{F}$ and $x_{\ell}=T$:
    \begin{itemize}
        \item $\displaystyle \big\lceil \sum_{q\in \faml_{\cT}(\ell)} \nu^\circ_{a}(q)- |\child_{\cT}(\ell)|\cdot \chi^\circ(P_{\ell}\setminus \sigma(t))\big\rceil\leq a^\circ$
        \item $\displaystyle \big\lceil \sum_{q\in \faml_{\cT}(\ell)} \nu^\circ_{b}(q) - |\child_{\cT}(\ell)|\cdot \chi^\circ(X\setminus \sigma(t)) \big\rceil\leq b^\circ$ 
    \end{itemize}
        \item If $x_A=\textsf{F}$ and $x_{\ell}=F$:
    \begin{itemize}
        \item $\displaystyle \big\lceil \sum_{q\in \faml_{\cT}(\ell)} \nu^\circ_{a}(q)\big\rceil\leq a^\circ$
        \item $\displaystyle \big\lceil \sum_{q\in \faml_{\cT}(\ell)} \nu^\circ_{b}(q) - |\child_{\cT}(\ell)|\cdot \chi^\circ(X\setminus \sigma(t)) - |\child_{\cT}(\ell)|\cdot \chi^\circ(P_{\ell}\setminus \sigma(t))\big\rceil\leq b^\circ$ 
    \end{itemize}
    \item Let $A_m=\sigma(t')$ and $A_{t'}\subseteq \sigma(t')$ such that:
    \begin{itemize}
        \item $\sigma(t')\cap X \subseteq A_{t'}$ if $x_A=T$ and $\sigma(t')\cap X = \emptyset$ otherwise
        \item $\sigma(t')\cap P_{\ell} \subseteq A_{t'}$ if $x_\ell=T$ and $\sigma(t')\cap P_{\ell} = \emptyset$ otherwise
        \item  $A_{t} \cap \sigma(t') = \sigma(t) \cap A_{t'}$ 
    \end{itemize} 
    then $M^X_{t'}(\nu_w(m),A_{t'},\nu^\circ_{a}(m),\nu^\circ_{b}(m))\neq \bot$,
    \item For each $q\in \child(\ell)$, $H(q,\nu_w(q),A_{t},\nu^\circ_{a}(q),\nu^\circ_{b}(q),x_A,x_\ell)\neq \bot$
\end{itemize}

\end{definition}
For $H(m,w,A_t,a^\circ,b^\circ,x_A,x_{\ell})$, if there is no feasible $(m,w,A_t,a^\circ,b^\circ,x_A,x_{\ell})$-configuration, we set it to $\bot$. If there is such a configuration, we build and set it to the following cut. 
We obtain a cut $(A^*,B^*)$ by combining the $A$-side and $B$-side of the cuts, $M^X_{t'}(\nu_w(m),A_{t'},\nu^\circ_{a}(m),\nu^\circ_{b}(m))$ and $H(q,\nu_w(q),A_{t},\nu^\circ_{a}(q),\nu^\circ_{b}(q),x_A,x_\ell)$, $\forall q\in \child(\ell)$. By definition of a feasible $(m,w,A_t,a^\circ,b^\circ,x_A,x_{\ell})$-configuration, all these cuts are $\neq \bot$. It is also easy to see that such a cut can indeed be formed. This is because the only common vertices are from $X$ and $P_{\ell}$ which are on fixed side $A$ or $B$ according to $x_A$ and $x_\ell$ and the cuts share no edges between them. We set $H(m,w,A_t,a^\circ,b^\circ,x_A,x_{\ell})=(A^*,B^*)$.

To complete the proof we need to show that the table $H$ constructed is good. 
Using induction, we can assume that for each $q\in \child(l)$, $H(q,.)$ is good and $M_t'$ is good. These properties along with the definition of $H$ and the notion of feasible configuration, can be used to verify that $H(m,.)$ is good.  One thing to note here is that a multiplicative rounding error of $(1+\delta)^2$ is introduced during this process which is reflected in properties $H\rightarrow h$ and $h\rightarrow H$. This completes our proof.


    
    

    
    

    
    




\section{W[1]-Hardness of \fb Parameterized by the Number of Colors} \label{sec:hardness}
In this section, we establish the \WO-hardness of \fb. To this end, we first consider the following problem.

\begin{tcolorbox}[colback=white!5!white,colframe=gray!75!black]
	\mdss
	\\\textbf{Input:} An instance $(\cV, T)$, where 
	\begin{itemize}
		\item $\cV = \{V_1, \ldots, V_n\}$, such that each $V_i \in \cV$ is a $d$-dimensional vector, i.e., $V_i \in \bbZp^d$.
		\item $T \in \bbZp^d$ is the $d$-dimensional \emph{target vector}.
	\end{itemize}
	\textbf{Question:} Does there exist a subset $U \subseteq \cV$ such that $\sum_{V_i \in \cU} V_i = T$?
\end{tcolorbox}
Although it is folklore that \mdss (MDSS) is \WOH parameterized by the dimension $d$, we are unable to find a reference for this result. Thus, we give a reduction from \csp to MDSS, where the former problem defined as follows. 

\begin{tcolorbox}[colback=white!5!white,colframe=gray!75!black]
	\csp
	\\\textbf{Input:} An instance $(X, \cC)$, where
	\begin{itemize}
		\item $X = \{x_1, x_2, \ldots, x_k\}$ is a set of $k$ variables.
		\item $\cC$ is a family of \emph{constraints}, where each constraint $C_{i, j} \in \cC$ corresponds to some $(x_i, x_j) \in \binom{X}{2}$, such that $C_{i, j} \in [n] \times [n]$ is the set of \emph{allowed pairs}.
	\end{itemize}
	\textbf{Question:} Does there exist an assignment $f: X \to [n]$, such that for each $C_{i, j} \in \cC$ corresponding to $(x_i, x_j)$, it holds that $(f(x_i), f(x_j)) \in C_{i, j}$?
\end{tcolorbox}
We have the following result from \cite{Marx10,LokshtanovR0Z20}. 

\begin{proposition}[\cite{Marx10,LokshtanovR0Z20}] \label{prop:csp-hard}
	\csp is \WO-hard parameterized by $k = |X|$, even in the special case where each variable $x_i \in X$ appears in $2$ or $3$ constraints. 
\end{proposition}

Our reduction from \csp (BCSP) to MDSS shows that the latter problem is \WO-hard even when the integer entries in each vector are bounded by a polynomial in $n$. Then, we reduce MDSS to \fb in two steps. As the first step, given an instance $(\cV, T)$ of MDSS, we reduce it to a special case, which we call \textsc{Multi-Dimensional Partition} (MDP), where the target vector $T$ is exactly half of the sum of entries along each dimension. Then, we reduce MDP to \fb, showing the \WO-hardness of the latter problem.

\subsection*{Reduction from BCSP to MDSS} Given an instance $(X, \cC)$ of \csp satisfying the property from \Cref{prop:csp-hard}, we reduce it to an instance of \mdss as follows. For each $x_i \in X$, let $d_i = |\cC(x_i)|$ be the number of constraints $x_i$ appears in. Note that $d_i \in \LR{2, 3}$. Let $d = \sum_{i = 1}^k d_i$, and note that $d \le 3k$. Let us define three large integers as follows: $N = 100n, A = 60n$ and $B = 40n$.

Now, we will define a set of $d$-dimensional vectors. To this end, we number the dimensions from $1$ to $d$. For each pair $(x_i, C_{i, j})$, where $C_{i, j} \in \cC(x_i)$, we associate a distinct dimension from $[d]$, which we identify by the pair $(x_i, C_{i, j})$. Note that this association is a bijection. Now, for each $x_i \in X$, and each $a \in [n]$, we add a vector $V(x_i, a)$ that has entries $A+a$ in the dimensions associated with pairs $(x_i, \cdot)$, and $0$ everywhere else. Note that the vector $V(x_i, a)$ has either two or three non-zero entries. 


Now, we add a few more vectors to $\cV$. For each $C_{i, j} \in \cC$, and each allowed pair $(a, b)$, such that the first coordinate corresponds to $x_i$ and second coordinate corresponds to $b$, we add a vector $V(C_{i, j}, a, b)$ to $\cV$, as follows. The vector $V(C_{i, j}, a, b)$ has entry $B-a$ in the dimension corresponding to $(x_i, C_{i, j})$, entry $B-b$ in the dimension corresponding to $(x_j, C_{i, j})$, and $0$ everywhere else. Note that the vectors in $\cV$ have non-negative entries. Finally, the target vector $T$ is defined as the $d$-dimensional vector with $N = A+B$ in each dimension.

\textbf{Forward direction.} Suppose $(X, \cC)$ is a yes-instance, then let $f: X \to \cC$ be an assignment that is compatible with every constraint. We show that there exists a subset $\cU \subseteq \cV$ such that $\sum_{V \in \cU} V = T$. For each $(x_i, C_{i, j})$, we add the vector $V(x_i, f(x_i))$ to $\cU$. Furthermore, for each $C_{i, j} \in \cC$, we add the vector $(C_{i, j}, f(x_i), f(x_j))$ to $\cU$. Now, we claim that the vectors in $\cU$ add up to $T$. Indeed, consider the dimension corresponding to $(x_i, C_{i, j})$, and note that only the vectors of the form $V(x_i, \cdot)$ and $V(C_{i, j}, \cdot, \cdot)$ can have non-zero entries in these dimensions. Now, $\cU$ contains only two such vectors: namely $V_1 = V(x_i, f(x_i))$, and $V_2 = V(C_{i, j}, f(x_i), f(x_j))$. Recall that $V_1$ and $V_2$ have entries $A+f(x_i)$ and $B-f(x_i)$ in this dimension, which implies that their sum is $A+B = N$. This concludes the forward direction.

\textbf{Reverse direction.} Let $\cU \subseteq \cV$ be a subset of vectors that add up to $T$. We first claim that, for every $x_i$, the set $\cU$ must contain exactly one vector of the form $V(x_i, \cdot)$. Indeed, consider the dimension corresponding to $(x_i, C_{i, j})$ for some $C_{i, j} \in \cC(x_i)$, and note that the entries in the vectors in $\cU$ add up to exactly $N$ along this dimension. We consider several cases.
\begin{enumerate}
	\item If $\cU$ contains at least $2$ vectors of the form $V(x_i, \cdot)$, then the sum of the entries in this coordinate is at least $2A > N$. Since all vectors have non-negative coordinates, this is a contradiction.
	\item If $\cU$ contains no vector of the form $V(x_i, \cdot)$, then we consider different cases. If $\cU$ contains at most two vectors of the form $V(C_{i, j}, \cdot, \cdot)$, then their sum in this coordinate is at most $2B < N$, which is a contradiction. Suppose $\cU$ contains $q \ge 3$ vectors of the form $V(C_{i, j}, \cdot, \cdot)$, then their sum in this coordinate is $qB - t$, where $t \le qn$. Therefore, $qB - t \ge = 39n \cdot q > N$, for any $q \ge 3$, which is also a contradiction.
\end{enumerate}
Thus, $\cU$ contains exactly one vector of the form $V(x_i, \cdot)$. Now, we show that for each pair $(x_i, C_{i, j})$, it also contains exactly one vector of the form $V(C_{i, j}, \cdot, \cdot)$. Again, consider the dimension corresponding to $(x_i, C_{i, j})$, and consider different cases.
\begin{enumerate}
	\item If $\cU$ contains zero vectors of the form $V(C_{i, j}, \cdot, \cdot)$, then the sum of the entries along this dimension is $A + a$, for some $0 \le a \le n$. However, $N - A = 60n \gg n \ge a$, which is a contradiction.
	\item If $\cU$ contains $q \ge 2$ vectors of the form $V(C_{i, j}, \cdot, \cdot)$, then the sum of the entries along this dimension is at least $A + a + qB - qn = A + qB - qn = A + q(39n) > N$ for any $q \ge 2$. This is also a contradiction. 
\end{enumerate}
Thus, $\cU$ contains exactly one vector of the form $V(x_i, \cdot)$, say $(x_i, a)$. We set $f(x_i) = a$.
This vector contains $A+a$ in the dimension corresponding to $(x_i, C_{i, j})$ for each $C_{i, j} \in \cC(x_i)$. For this $C_{i, j}$, $\cU$ also contains exactly one vector of the form $V(C_{i, j}, \cdot, \cdot)$. It follows that this vector must be $V(C_{i, j}, a, \cdot)$, since the entries along the dimension corresponding to $(x_i, C_{i, j})$ must add up to $A+B$. Therefore, $f(x_i) = a$ is compatible with $C_{i, j}$. It follows that the instance $(X, \cC)$ of \csp admits a satisfying assignment. 

We conclude with the following result.

\begin{proposition}\label{prop:subsetsum-hard}
	\mdss is \WO-hard parameterized by $d$, the number of dimensions, even with the following assumptions:
	\begin{enumerate}
		\item All the entries in the vectors are non-negative integers bounded by a polynomial in the number of vectors,
		\item Target vector is $c \mathbf{1}$ for some positive integer $c$, which is bounded by a polynomial in $n$, such that $(2c) \cdot \mathbf{1} \le \sum_{V \in \cV} V$.
	\end{enumerate}
\end{proposition}

\subsection*{Reduction from MDSS to \fb} Let $(\cV, T)$ be an instance satisfying the properties stated in \Cref{prop:subsetsum-hard}. Now, we reduce \mdss to \mdp, a special case of \mdss, where the target vector is exactly the half of the sum of each coordinate. To achieve this, first we add a vector $V'$ to $\cV$, whose entries are chosen from $\{c+1, c+2\}$ of appropriate parity, such that all the entries in the vector $S$ obtained by adding all the vectors in (the new) $\cV$, are even. Note the addition of $V'$ does not change a yes-instance to a no-instance or vice versa, and properties 1 and 2 from \Cref{prop:subsetsum-hard} continue to hold. Now, we add a vector $S - 2T \ge \mathbf{0}$ to $\cV$, and change the target vector to $S/2$. It is straightforward to see that the resulting instance is equivalent to the original instance of \mdss, which we refer to as an instance of \mdp.

Let $\cV$ be the set of vectors that are input to \mdp (note that the target vector is automatically defined). We give a parameter-preserving reduction to \fb, as follows. For a vector $V_i \in \cV$, let $V_i[j]$ denote the entry in the $j$th coordinate, where $1 \le j \le d$. Note that $V_i[j]$ is a non-negative integer. Let $s_i = \sum_{j = 1}^d V_i[j]$. For the vector $V_i$, we proceed as follows. For each $1 \le j \le d$, let $C_{i, j}$ be a set of $V_i[j]$ many distinct vertices of \emph{color} $j$. Let $C_i = \bigcup_{j = 1}^d C_{i, j}$. We select an arbitrary vertex $u \in C_i$, and connect all the remaining vertices in $C_i$ to $u$ such that the resulting graph $G_i$ is a star centered at $u$. For each $V_i \in \cV$, we create graphs $G_i$ in this manner. Finally, let $G$ denote the disjoint union of all graphs $G_i$ defined in this manner. It is easy to see that a fair bisection of cut-size $0$ exists iff $\cV$ is a yes-instance of \mdp. This establishes the following theorem.

\begin{restatable}{theorem}{hardnessthm}\label{thm:fair-bisection-hardness}
	\fb is \WOH parameterized by the number of colors $c$, even when $k$, the cut-size is zero. 
\end{restatable}

\section{Conclusion} \label{sec:conclusion}
In this paper, we initiated the study of \fb from the perspective of parameterized algorithms. We showed that the problem is \WOH parameterized by the number of colors $c$, even when $k = 0$; thus, we cannot hope to generalize the \FPT algorithm to \fb with a running time of the form $f(k, c) \cdot n^{\cO(1)}$. On the other hand, the known $2^{\cO(k \log k)} \cdot n^{\cO(1)}$ algorithm for \mb (\cite{CyganKLPPSW21,CyganLPPS19}) extends to \fb in a straightforward manner with running time $2^{\cO(k \log k)} \cdot n^{\cO(c)}$. Our main result is that \fb admits an \FPT{}-approximation algorithm 
that finds an $(\epsilon, {\sf r})$-fair bisection in time $2^{\cO(k \log k)} \cdot \lr{\frac{c}{\epsilon}}^{\cO(c)} \cdot n^{\cO(1)}$. In fact, by setting $\epsilon = 1/(2n)$, we can obtain the previously mentioned exact algorithm as a corollary.

We note that our approximation algorithm also works in the setting where a vertex can belong to multiple color classes. Furthermore, our technique can be extended to \textsc{Fair $q$-section} problem, where we want to partition the vertex set into $q$ parts such that (i) at most $k$ edges with endpoints in different parts, and (ii) each part has proportional representation from each color -- here, the algorithm will have an \textsf{XP} dependence on $q$. 

Our main conceptual contribution is the observation that it is possible to design parameterized approximation algorithms by applying the technique of Lampis~\cite{Lampis14} to dynamic programming algorithms over tree decompositions with unbreakable bags. Towards this goal we designed an algorithm that given a graph $G$ and integer $k$ computes a $(9k, k)$-unbreakable tree decomposition of $G$ with logaritmic depth and adhesions of size at most $8k$ in time $2^{\cO(k \log k)} n^{\cO(1)}$. We expect that this will be a useful tool for obtaining parameterized approximation algorithms for other problems by using Lampis~\cite{Lampis14}-style dynamic programming over tree decompositions with unbreakable bags.



\bibliography{refs}



\end{document}